\documentclass[11pt]{article}
\usepackage{amsmath,amsthm,nicefrac,amssymb}
\usepackage{amsfonts, amstext}

\usepackage[utf8]{inputenc}
\usepackage{enumitem}
\usepackage{geometry}
\usepackage[ruled,vlined,linesnumbered]{algorithm2e}
\usepackage{algpseudocode}
\usepackage{amsmath,amsfonts,amsthm,amssymb,dsfont}
\usepackage{color}
\usepackage{graphicx}

\usepackage{thmtools}

\usepackage{xcolor}
\usepackage{nameref}
\definecolor{ForestGreen}{rgb}{0.1333,0.5451,0.1333}
\definecolor{DarkRed}{rgb}{0.65,0,0}
\definecolor{Red}{rgb}{1,0,0}
\usepackage[linktocpage=true,
pagebackref=true,colorlinks,
linkcolor=DarkRed,citecolor=ForestGreen,
bookmarks,bookmarksopen,bookmarksnumbered]{hyperref}
\usepackage{cleveref}

\newcommand{\tO}{\tilde{O}}
\newcommand{\tOmega}{\tilde{\Omega}}
\newcommand{\alert}[1]{{\textcolor{red}{#1}}}

\newcommand{\eat}[1]{}

\ifdefined\ShowComment
\def\thatchaphol#1{\marginpar{$\leftarrow$\fbox{T}}\footnote{$\Rightarrow$~{\sf\textcolor{purple}{#1 --Thatchaphol}}}}
\def\jason#1{\marginpar{$\leftarrow$\fbox{JL}}\footnote{$\Rightarrow$~{\sf\textcolor{blue}{#1 --Jason}}}}

\def\danupon#1{{\sf\textcolor{green}{DN: #1}}}

\def\kent#1{\marginpar{$\leftarrow$\fbox{K}}\footnote{$\Rightarrow$~{\sf\textcolor{red}{#1 --Kent}}}}

\def\ruoxu#1{\marginpar{$\leftarrow$\fbox{R}}\footnote{$\Rightarrow$~{\sf\textcolor{cyan}{#1 --Ruoxu}}}}

\def\debmalya#1{\marginpar{$\leftarrow$\fbox{DP}}\footnote{$\Rightarrow$~{\sf\textcolor{brown}{#1 --Debmalya}}}}

\else

\def\thatchaphol#1{}
\def\jason#1{}
\def\danupon#1{}
\def\kent#1{}
\def\ruoxu#1{}
\def\debmalya#1{}

\fi

\declaretheorem[numberwithin=section]{theorem}
\declaretheorem[numberlike=theorem]{lemma}

\crefname{algorithm}{Algorithm}{Algorithms}
\Crefname{algorithm}{Algorithm}{Algorithms}

\theoremstyle{definition}
\declaretheorem[numberlike=theorem]{definition}

\usepackage{xparse}             %
\usepackage{xspace}              %
\usepackage[mathscr]{euscript} %

\usepackage{mleftright}         %
\usepackage{mathbbol}           %
\usepackage{fifo-stack}
\usepackage{thmtools}
\usepackage{thm-restate}
\usepackage{letltxmacro}        %
\usepackage{xpatch}              %

\usepackage[T1]{fontenc}%
\usepackage[utf8]{inputenc}%
\usepackage{xparse}
\usepackage{framed}
\usepackage{comment}

\usepackage[export]{adjustbox}
\usepackage{xparse} %
\usepackage{environ} %
\usepackage{xstring}
\usepackage{tabularx}
\usepackage{enumitem}
\usepackage[htt]{hyphenat}%
\usepackage{varwidth}     %
\usepackage{needspace}          %

\input{krq/math.sty}            %
\input{krq/algo.sty}            %
\input{krq/article.sty}         %

\newenvironment{krq}{
  \begingroup
}{
  \endgroup
}

\providecommand{\outdegree}{\fparnew{\operatorname{deg}^+}} %

\NewDocumentCommand{\cutsize}{O{\delta} g g e{_}}{%
  #1
  \IfNoValueF{#3}{_{#3}}%
  \IfNoValueF{#4}{_{#4}}%
  \IfNoValueF{#2}{\parof{#2}}
}%

\newcommand{\incutsize}{\cutsize[\delta^-]}%
\newcommand{\outcutsize}{\cutsize[\delta^+]}%
\newcommand{\outcut}{\cutsize[\partial^+]}%
\newcommand{\incut}{\cutsize[\partial^-]}%

\newcommand{\inneighbors}{\fparnew{N^-}}  %
\newcommand{\outneighbors}{\fparnew{N^+}} %

\providecommand{\ec}{\lambda} %
\providecommand{\vc}{\kappa} %
\providecommand{\inneighbors}{\fparnew{N^{-}}} %
\providecommand{\inneighbors}{\fparnew{N^-}} %

\providecommand{\sparseG}{G_0}
\providecommand{\sparseV}{V_0}
\providecommand{\sparseE}{E_0}

\providecommand{\weight}{\fparnew{w}} %
\providecommand{\revG}{G_{\text{rev}}} %
\providecommand{\splitG}{G_{\text{split}}}%
\providecommand{\tout}{t^+}    %
\providecommand{\rin}{\root^-}          %
\providecommand{\vin}{v^-}
\providecommand{\vout}{v^+}         %
\providecommand{\uout}{u^+}
\providecommand{\optSink}{T^{\star}} %
\providecommand{\Sink}{T} %
\newcommand{\defterm}[1]{\emph{#1}} %
\providecommand{\citet}{\cite}%
\newcommand{\sink}{t}%
\renewcommand{\root}{s}%

\newcommand{\optSource}{S^*}%
\renewcommand{\optSink}{T^*}%
\newcommand{\varec}{\tilde{\ec}}%

\newcommand{\val}{\mathrm{val}}

\renewcommand{\paragraph}[1]{\medskip\noindent{\bf #1}\xspace}

\title{Minimum Cuts in Directed Graphs via Partial Sparsification\thanks{This paper combines, and improves on, two independent manuscripts by Quanrud~\cite{Quanrud21} and the other authors~\cite{CenLNPS21}.}}
\author{Ruoxu Cen\\Duke University\and Jason Li\\Carnegie Mellon University\and Danupon Nanongkai\\University of Copenhagen \& KTH \and Debmalya Panigrahi\\Duke University \and Kent Quanrud\\Purdue University\and Thatchaphol Saranurak\\University of Michigan, Ann Arbor}

\begin{document}

\maketitle

\eat{

\begin{abstract}
We give an algorithm to find a minimum cut in an edge-weighted directed graph with $n$ vertices and $m$ edges in $\tO(n\cdot \max(m^{2/3}, n))$ time. This improves on the 30 year old bound of $\tO(nm)$ obtained by Hao and Orlin for this problem. Our main technique is to reduce the directed mincut problem to $\tO(\min(n/m^{1/3}, \sqrt{n}))$ calls of {\em any} maxflow subroutine. Using state-of-the-art maxflow algorithms, this yields the above running time.
Our techniques also yield fast {\em approximation} algorithms for finding minimum cuts in directed graphs. For both  edge and vertex weighted graphs, we give $(1+\epsilon)$-approximation algorithms that run in $\tO(n^2 / \epsilon^2)$ time.
\end{abstract}

}

\begin{abstract}
We give an algorithm to find a minimum cut in an edge-weighted directed graph with $n$ vertices and $m$ edges in $\tO(n\cdot \max\{m^{2/3}, n\})$ time. This improves on the 30 year old bound of $\tO(nm)$ obtained by Hao and Orlin for this problem.  Using similar techniques, we also obtain $\tO(n^2 / \eps^2)$-time $(1+\eps)$-approximation algorithms for both the minimum edge and minimum vertex cuts in directed graphs, 
for any fixed $\epsilon$. Before our work, no $(1+\eps)$-approximation algorithm better than the exact runtime of $\tO(nm)$ is known for either problem.

Our algorithms follow a two-step template. In the first step, we employ a {\em partial sparsification} of the input graph to preserve a critical subset of cut values approximately. In the second step, we design algorithms to find the (edge/vertex) mincut among the preserved cuts from the first step. For edge mincut, we give a new reduction to $\tO(\min\{n/m^{1/3}, \sqrt{n}\})$ calls of \emph{any} maxflow subroutine, via packing arborescences in the sparsifier. For vertex mincut, we develop new local flow algorithms to identify small unbalanced cuts in the sparsified graph.

\end{abstract}

\pagenumbering{gobble}
\clearpage
\tableofcontents

\clearpage

\pagenumbering{arabic}

\section{Introduction}
\label{sec:introduction}

The {\em minimum cut} (or mincut) problem is one of the most widely studied problems in graph algorithms. In (edge-)weighted\footnote{We assume throughout that edge/vertex weights are polynomially bounded integers.} directed graphs (or {\em digraphs}), a mincut is a bipartition of the vertices into two non-empty sets $(S, V\setminus S)$ so that the total weight of edges from $S$ to $V\setminus S$ is minimized. 
This problem can be solved by solving the {\em $s$-mincut} problem (also called rooted mincut), where for a given {\em root} vertex $s$, we want to find the minimum weight cut $(S, V\setminus S)$ such that $s\in S$. (We call such cuts minimum $s$-cuts or $s$-mincuts.) This is because the mincut can be computed as the minimum between two $s$-mincuts for an arbitrary vertex $s$: one with the original edge directions in the input digraph, and the other with the edge directions reversed.


A simple algorithm for $s$-mincut (and thus mincut) on an $m$-edge, $n$-vertex digraph is to use $n-1$ maxflow calls to obtain the minimum $s-t$ cut for every vertex $t \not= s$ in the graph, and return the minimum among these. A beautiful result of Hao and Orlin~\cite{HaoO92} showed that these maxflow calls can be amortized to match the running time of a single maxflow call, provided one uses the  {\em push-relabel maxflow algorithm}~\cite{GoldbergT88}. This leads to an overall running time of $\tO(mn)$. 
Since their work, better maxflow algorithms have been designed, but the amortization does not work for these algorithms. 
As a consequence, the Hao-Orlin bound remains the best known for the directed mincut problem almost 30 years after their work.

\danupon{This is too late now, but perhaps in the next version we can consider discussing undirected graphs and unwieghted graphs so the readers see the full picture. We should also discuss why it's much more challenging to attack weighted directed graphs.}

\eat{
Using a different technique of duality between rooted mincuts and arborescences, Gabow~\cite{Gabow1995} obtained a running time of $\tO(m\lambda)$ for $s$-mincut, where $\lambda$ is the weight of a mincut (assuming integer weights).

%
For simple, unweighted graphs, we know $\lambda \le n-1$,
so Gabow's algorithm is no worse than Hao-Orlin, and sometimes better. Recently \cite{cq-simple-connectivity} obtained a randomized $\tO(n^2 U)$ time algorithm for integer capacities between $1$ and $U$.
For arbitrary weights, however, the the Hao-Orlin bound of $\tO(mn)$ remains the state of the art for the directed mincut problem.

\kent{I wrote the following to brief introduce vertex connectivity. I tried to keep it short. Otherwise it shows up for the first time in the results section which seems a little abrupt.}
We also consider vertex connectivity in vertex-weighted digraphs. A {\em vertex cut} in a digraph is a set of vertices $X \subset V$ such that removing $X$ from the graph ensures that $G$ is not strongly connected (or consisting of a single vertex). The \emph{minimum vertex cut} (or vertex mincut) problem is to find the vertex cut of minimum weight. One also has the rooted version of the problem where there is a fixed vertex $s$, and the goal is to compute the minimum weight set of vertices $X \subseteq V - \root$ such that either $X = V - \root$ or, after removing $X$, $\root$ is unable to reach at least one other vertex. Currently the fastest exact algorithm for minimum global and rooted vertex cut in weighted graphs is due to Henzinger, Rao, and Gabow \cite{hrg-00} and runs in $\apxO{m n}$ randomized time. Interestingly this algorithm is based on the Hao-Orlin algorithm mentioned above. This bound has recently been improved for the special case of unweighted digraphs to $\tO(mn^{11/12+o(1)})$ \cite{li+21}. 
This result is one of several recent developments for vertex connectivity in unweighted graphs as discussed below in \Cref{related-work}.

}

\subsection{Our Results}
In this paper, we can solve the $s$-mincut---thus the directed mincut problem---by essentially reducing it to $O(\sqrt{n})$ maxflow calls. At first glance, this is worse than the Hao-Orlin algorithm that only uses a single maxflow call. But crucially, while the Hao-Orlin algorithm is restricted to a specific maxflow subroutine and therefore cannot take advantage of faster, more recent maxflow algorithms, our new algorithm treats the maxflow subroutine as a black box, thereby allowing the use of {\em any} maxflow algorithm. Using state of the art maxflow algorithms that run in $\tO(m+n^{3/2})$ time~\cite{BrandLLSSSW21}, this already improves on the Hao-Orlin bound. 
Using some additional ideas, we further reduce to $O(\min\{n/m^{1/3}, \sqrt{n}\})$ maxflow calls, which yields our eventual running time of $\tO(nm^{2/3} + n^2)$:

\begin{theorem}
\label{thm:main}
There is a randomized Monte Carlo algorithm that solves $s$-mincut (and therefore directed mincut) whp in $\tO(nm^{2/3} + n^2)$ time on an $n$-vertex, $m$-edge (edge-weighted) directed graph.
\end{theorem}
\noindent
%
%
%
In fact, our reduction in general implies 
a running time bound of $\tO(\min\{mk, nk^2\}+\frac{n}{k}\cdot F(m,n))$, where $k$ is a parameter that we can choose and $F(m,n)$ is the time complexity of maxflow (see \Cref{{thm:main-primitive}}).

Our techniques also yield fast approximations for the mincut problem in directed graphs. In particular, for any $\epsilon\in(0, 1)$, we can find a $(1+\epsilon)$-approximate mincut in $\tO(n^2/\epsilon^2)$ time:

\begin{theorem}
\label{thm:edge-apx} 
  For any $\epsilon \in (0,1)$, there is a randomized Monte Carlo algorithm that finds a $(1+\epsilon)$-approximate $s$-mincut (and therefore directed mincut) whp in $\tO(n^2 / \epsilon^2)$ time on an $n$-vertex (edge-weighted) directed graph.
\end{theorem}

Finally, we consider vertex-weighted digraphs. A {\em vertex cut} in a digraph is defined as a tri-partition of vertices into sets $(L,X,R)$ such that there is no edge from $L$ to $R$. \danupon{Should we say $L, R \not= \emptyset$?} (In other words, removing the vertices in $X$ results in a digraph where the directed cut $(L, R)$ is empty.) A {\em minimum} vertex cut (or vertex mincut) is a vertex cut $(L, X, R)$ that minimizes the sum of weights of vertices in $X$. We give an algorithm to find a $(1+\epsilon)$-approximate vertex mincut in $\tO(n^2/\epsilon^2)$ time:

\begin{theorem}
\label{thm:vertex-apx}
  For any $\epsilon \in (0,1)$, there is a randomized Monte Carlo algorithm that finds a $(1+\epsilon)$-approximate minimum vertex $s$-cut and the minimum global vertex cut whp in $\tO(n^2 / \epsilon^2)$ time on an $n$-vertex (vertex-weighted) directed graph.
\end{theorem}

To the best of our knowledge, before this work, the fastest algorithms for $(1+\epsilon)$-approximation of mincuts in edge or vertex weighted directed graphs were the respective exact algorithms themselves, which obtained a running time of $\tO(mn)$~\cite{HaoO92,hrg-00}.
 Our approximation results establish a separation between the best exact and $(1+\epsilon)$-approximation algorithms for both edge and vertex mincut problems in directed graphs.

\smallskip
\noindent
{\em Remark:}
Our results are the first to
break the $O(mn)$ barrier for directed mincut problems in general, weighted digraphs. For all values of $m$ except when $m = n^{1+o(1)}$, this is immediate from the above theorems. If $m = n^{1+o(1)}$, we can also break the $O(mn)$ barrier by employing the recent $\apxO{m^{1.5 - 1/328}}$-time max-flow algorithm of $\cite{glp-21}$
to obtain $\bigO{mn^{1-\Omega(1)}}$-time algorithms for all problems in \Cref{thm:main,thm:edge-apx,thm:vertex-apx}.

\smallskip
\noindent   
{\bf Related Work.}
The directed (edge) mincut problem has been studied over several decades. Early work focused on unweighted graphs~\cite{EvenT75,Schnorr79} eventually resulting in an $O(mn)$-time algorithm due to Mansour and Schieber~\cite{MansourS89}. This was matched (up to log factors) in weighted graphs by Hao and Orlin~\cite{HaoO92}, whose result we improve in this paper. For unweighted graphs (and graphs with small integer weights), the current record is a recent $\tO(n^2)$-time algorithm due to Chekuri and Quanrud~\cite{cq-simple-connectivity}. A similar story has unfolded for directed vertex mincuts. Early work again focused on unweighted graphs~\cite{podderyugin-73,EvenT75,cheriyan-reif,galil-80} until the work of Henzinger, Rao, and Gabow~\cite{hrg-00} who obtained an $\tO(mn)$-time algorithm for weighted graphs. The current best for directed vertex mincut in unweighted graphs is an $\tO(m n^{11/12+o(1)})$-time algorithm due to Li~{\em et al.}~\cite{li+21}. Faster algorithms are known when the mincut size is small and for $(1+\epsilon)$-approximations in unweighted digraphs~\cite{NanongkaiSY19,forster+20,cq-simple-connectivity}. 
%

\subsection{Our Techniques}
Our results are obtained by solving the $s$-mincut problem. Let us consider the edge-weighted case. Gabow~\cite{Gabow1995} obtained a running time of $\tO(m\lambda)$ for this problem (assuming integer weights), where $\lambda$ is the size of an $s$-mincut. He did so via {\em arborescense packing}: Define an $s$-arborescense to be any spanning tree rooted at $s$ with edges pointing toward the leaves.
In $\tO(m\lambda)$ time, Gabow's algorithm computes $\lambda$ $s$-arborescenses such that an edge $e$ of weight $w(e)$ is contained in at most $w(e)$ arborescenses (this is called arborescense packing).\footnote{Gabow actually constructs a directionless spanning tree packing, which is a relaxation of an arborescence packing, but we ignore this technical detail here since it is not relevant to our eventual algorithm.} 
%
%
Gabow's algorithm is at least as fast as that of Hao and Orlin for unweighted simple graphs (since $\lambda \le n-1$), but can be much worse for weighted (or multi) graphs. Nevertheless, Karger~\cite{Karger00} gave an interesting approach to use arborescence packing for the mincut problem even with edge weights, {\em but only in undirected graphs}. Karger's algorithm had three main steps:
\begin{itemize}[noitemsep]
    \item[(a)] sparsify the input graph $G$ to $H$ by random sampling of edges to reduce the mincut value in $H$ to $O(\log n)$ while guaranteeing that the mincut in $G$ is a $(1+\epsilon)$-approximate mincut in $H$,
    \item [(b)] pack $O(\log n)$ $s$-arborescences\footnote{Since Karger's algorithm considered undirected graphs, the arborescences are simply spanning trees.} 
    in the sparsifier $H$, and 
    \item [(c)] find the minimum weight cut among those that have at most two edges in an arboresence using a dynamic program.
\end{itemize}

The last step is sufficient because the duality between cuts and arborescences ensures that an $s$-mincut, which is now a $(1+\epsilon)$-approximate mincut after sparsification, has at most two edges in at least one $s$-arborescence.\footnote{The duality implies that if the undirected mincut is $\lambda$, then we can pack $\lambda$ spanning trees where every edge appears in at most two arborescences.}
%
Karger implements all these steps in $\tO(m)$ time to obtain an $\tO(m)$-time mincut algorithm in undirected graphs.

Unfortunately, steps (a) and (c) in Karger's scheme are not valid in a directed graph. 
To begin with, directed graphs do not admit sparsifiers similar to Karger's sparsifier: while Karger's sparsifier  approximately preserves all cuts in an undirected graph (after some scaling), it is well known that a sparsifier with a similar property does not exist in directed graphs (see, e.g., \cite{CenCPS21}).
This is mainly because we cannot bound the number of mincuts and approximate mincuts in directed graphs while we can do so in undirected graphs. 


\paragraph{Partial sparsification.} 
Since it is impossible for a sparsifier to preserve all cuts in directed graphs, it is natural to try to preserve a {\em partial} subset of cuts. 
%
%
%
%
Suppose we were guaranteed that the $s$-mincut $(S, V\setminus S)$ (recall that $s\in S$) is unbalanced in the sense that $|V\setminus S| \leq k$ for some parameter $k$ that we will fix later. Let us randomly
sample edges to scale down the value of the mincut to $\tO(k)$. In an undirected graph, as long as $k = \Omega(\log n)$, all the cuts would converge to their expected values. This is not true in digraphs, but crucially, {\em all the unbalanced cuts converge to their expected values} since there are only $n^{\tO(k)}$ of them.
However, it is possible that some balanced cut is (misleadingly) the new mincut of the sampled graph, having been scaled down disproportionately by the random sampling. So, we {\em overlay} this sampled graph with a star rooted at $s$, and show that this sufficiently increases the values of all balanced cuts, while only distorting the unbalanced cuts slightly. After this overlay, we can claim that we now have a digraph where $(S, V\setminus S)$ is a $(1+\epsilon)$-approximate mincut. We use one additional idea here. We show that in the sparsifier, every vertex in $V\setminus S$ has only $\apxO{k}$ incoming edges (note that each edge can be weighted)---$\apxO{k}$ edges across the cut from $S$ and $\le k$ edges from within $V\setminus S$. By contracting all vertices with unweighted in-degree $> \apxO{k} $ into $s$, we reduce the number of edges in the digraph to $\apxO{nk}$.

But, what if our premise that the $s$-mincut is unbalanced does not hold? This case is actually simple. We use a uniform random sample of $\tO(\frac{n}{k})$ vertices, and find $\root-\sink$ mincuts for all vertices $\sink$ from the sample, using $\tO(\frac{n}{k})$ maxflow calls. It is easy to see that whp, the sample {\em hits} $V\setminus S$, and hence, the minimum weighted cut among these $s-t$ mincuts will reveal the $s$-mincut of the graph.

Let us, therefore, return to the case where the $s$-mincut is unbalanced. Recall that we have already sparsified the graph to one that has only $\tO(nk)$ edges, and where the mincut has weight $\tO(k)$. The next step is to create a maximum packing of edge-disjoint $s$-arborescences. Because the graph is weighted, instead of using Gabow's algorithm described above, we create a (fractional) packing using a multiplicative weights update procedure (e.g.,~\cite{Young1995}). By duality,\footnote{By duality, we have that if the directed mincut is $\lambda$, then we can pack $\lambda$ arborescenses where each edge appears in at most one arborescense. Note that this is different from the case of undirected graphs where each edge appears in at most two arborescences.} 
these arborescences have the following property: if we sample $O(\log n)$ random $s$-arborescences, then whp there will be at least one arborescence $T$ such that there is exactly one edge in $T$ that goes from $S$ to $V\setminus S$. 
%
In this case, we say that the cut $(S, V\setminus S)$ {\em $1$-respects} the arborescence $T$.

\paragraph{Our 1-respecting cut algorithm.}
Our final task, therefore, is the following: given an arborescence, find the minimum weight cut in the original graph among all those that $1$-respect the arborescence $T$. At first sight, this may look similar to part (c) of Karger's algorithm which can be solved by a dynamic program or other techniques (e.g.~\cite{MukhopadhyayN20,GawrychowskiMW20,GawrychowskiMW21,LopezMN}).
%
%
However, these techniques
relied on the fact that if the $s$-mincut $(S, V\setminus S)$ has a single edge $e$ in $T$, then 
$S$ and $V\setminus S$ would be contiguous in an $s$-arborescence $T$ (i.e. removing $e$ from $T$ gives $S$ and $V\setminus S$ as the two connected components). 
%
%
%
This is not true for a directed graph: While an $s$-arborescence will contain exactly one edge from $S$ to $V\setminus S$, it could contain an arbitrary number of edges in the opposite direction from $V\setminus S$ to $S$, thereby only guaranteeing the contiguity of $V\setminus S$ but not $S$.

One of the main contributions of this paper 
is to provide an algorithm to solve the above problem using $O(\log n)$ maxflow computations. The main idea is to use a centroid-based recursive decomposition of the arborescence, where in each step, we use a set of maxflow calls that can be amortized on the original graph. The minimum cut returned by all these maxflow calls is eventually returned as the $s$-mincut of the graph.

\paragraph{Approximation Algorithms.}
The ideas above also lead to a quadratic time approximation algorithm for edge mincuts. At a high level, if we execute each $(s,t)$-max flow in the sparsifier instead of in the input graph (both in the unbalanced and balanced settings, with care), then we obtain a $\tO(n^2/\eps^2)$ time $\epsmore$-approximation algorithm instead.

A similar approach can be taken for approximate vertex mincuts. Our partial sparsification technique reduces the graph to $\tO(nk/\epsilon^{2})$ edges while maintaining the vertex $s$-mincut (with some additional adjustments for vertex mincuts). For large $k$, we similarly run $(s,t)$-max flow between $\tO(n/k)$ pairs of vertices on the sparsifier. For small $k$, we design a new local cut algorithm from $\tO(n/k)$ seeds each of which takes $\tO(k^{3}/\epsilon^{2})$ time. 
This local algorithm is inspired by local algorithms for unweighted graphs~\cite{forster+20}, and speeds up the running time by a factor of $k$ by leveraging the structure of our sparsifiers (beyond sparsity). We finally obtain $\tO(n^{2}/\epsilon^{2})$ running time by balancing the two cases and calling the max flow algorithm by \cite{BrandLLSSSW21}.

\paragraph{Summary.}
To summarize, our algorithms distinguish between balanced and unbalanced mincuts, solving the former using maxflows on randomly sampled terminals. For unbalanced mincuts, we follow a two-step template. In the first step, we employ partial sparsification to preserve the values of unbalanced cuts approximately, while suppressing balanced cuts using an overlay. In the second step, we design algorithms to find the (edge/vertex) mincut among unbalanced cuts. For edge mincut, the sparsifier allows one to quickly obtain an arborescence that 1-respects the directed mincut. From this arboresence, we obtain the exact minimum cut in $O(\log n)$ max flows via a centroid-based recursive decomposition. For vertex mincut, we develop new local flow algorithms to identify small unbalanced cuts in the sparsified graph.

\eat{
\paragraph{Summary.}
To summarize, our main insights are to use sparsifiers with additional overlays and contractions for the unbalanced case, and to find the mincut that 1-respects an arborescence via maxlow and centroid-based recursive decomposition. 
Overall, we run the algorithms for both the balanced and unbalanced cases, and return the minimum cut returned by either algorithm. For the balanced case, we need $\tO(\frac{n}{k})$ maxflows. For the unbalanced case, we need $\tO(1)$ maxflows in a graph containing $O(nk)$ edges. All other steps run in $\tO(m)$ time. Balancing these gives $k = \max(m^{1/3}, \sqrt{n})$, thereby resulting in \Cref{thm:main}.

\paragraph{Summary (new).}\kent{Here I have tried to extend the summary to capture both the edge and vertex connectivity algorithms.}  At a high level, both the edge connectivity and vertex connectivity algorithms leverage the insight that when the minimum cuts are unbalanced, one can combine random sampling with contractions and other graph modifications to produce a smaller sparsifierq that approximately preserves the minimum cut. The sparsifier for edge connectivity allows one to quickly obtain an arborescence that 1-respects the minimum directed cut, from which the exact minimum cut can be obtained in $O(\log n)$ max flows via the new centroid-based recursive decomposition. Meanwhile, the sparsifier for vertex connectivity permits a substantially faster local algorithm for vertex cuts.  In the opposite regime, where the minimum cut is balanced,  the minimum cut can instead be obtained very quickly by randomly sampling a source and sink from opposite sides of the minimum cut, and running $(s,t)$-max flow.  (The quadratic time approximation algorithms also apply the sparsifiers in the balanced setting before running max flow.) This combination of approaches -- exploiting a tradeoff between better sparsifiers for more unbalanced minimum cuts with fewer source-sink samples for more balanced minimum cuts -- leads to the improved running time for directed minimum cut.

\kent{I suggest dropping the following paragraph, relying on a somewhat similar discussion in Section 2. Somewhere up above I placed a sentence saying: "If instead we execute each $(s,t)$-max flow in appropriately sparsified graphs (both in the unbalanced and balanced settings, with care), then we obtain am $\tO(n^2/\eps^2)$ time algorithm instead."}

\danupon{Did we hide $\log U$ anywhere (except when we use \cite{GaoLP21})?}

}


\eat{

\paragraph{Related Work.}
The mincut problem has been studied in directed graphs over several decades. Early work focused on unweighted graphs, starting with the work of Even and Tarjan in 1975~\cite{EvenT75}, followed by improvements by Schnorr~\cite{Schnorr79}, and eventually resulting in an $O(n\cdot \min(m, \lambda^2 n))$-time algorithm due to Mansour and Schieber~\cite{MansourS89}. This was matched in weighted graphs by Hao and Orlin~\cite{HaoO92}, whose result we improve in this paper. 

Note that the $\tO(mn)$ bound was recently improved in \emph{unweighted} directed graphs for both edge and vertex mincut problems. Chekuri and Quanrud  \cite{cq-simple-connectivity} give a $\tO(n^2)$-time algorithm for finding directed edge mincuts.\footnote{In fact, the algorithm works on graphs with integer capacities between $1$ and $U$ and takes $\tO(n^2 U)$ time} A key idea in this work is a \emph{deterministic} rooted sparsification technique that inspires some of the ideas in this work.  For directed vertex mincuts, there are algorithms with $\tO(m n^{11/12+o(1)})$ time \cite{li+21} and even  faster algorithms when the mincut size is small and/or when $(1+\epsilon)$-approximation is allowed \cite{NanongkaiSY19,forster+20,cq-simple-connectivity}. The techniques from these recent developments seem to heavily rely on the fact that the input graph is unweighted.


}

\eat{

To overcome these challenges, we adopt several ingredients that we outline below:
\begin{itemize}
    \item We consider two possibilities: either the mincut has $\tOmega(k)$ vertices on the smaller side or fewer (let us call these {\em balanced} and {\em unbalanced} cuts respectively). If the mincut is a balanced cut, we use two random samples of $\tO(\frac{n}{k})$ and $\tO(1)$ vertices each, and find $s-t$ mincuts for all pairs of vertices from the two samples. It is easy to see that whp, the two samples would respectively {\em hit} the smaller and larger sides of the mincut, and hence, one of these $s-t$ mincuts will reveal the overall mincut of the graph.
    \item The main task, then, is to find the mincut when it is unbalanced. In this case, we use a sequence of steps. The first step is to use {\em cut sparsification} of the graph by random sampling of edges. This scales down the size of the mincut, but unlike in an undirected graph, all the cuts of a digraph do not necessarily converge to their expected values in the sample. However, crucially, {\em the mincut can be scaled to $\tO(k)$ while ensuring that all the unbalanced cuts converge to their expected values.}
    \item Since only the unbalanced cuts converge to their expected values, it is possible that some balanced cut is the new mincut of the sampled graph, having been scaled down disproportionately by the random sampling. Our next step is to {\em overlay} this sampled graph with an star, 
    which increases all balanced $s$-cuts to be  sufficiently large, while the unbalanced cuts are only distorted by a small multiplicative factor. 
    \item In the sparsifier, all vertices in the sink side of a potential minimum $s$-cut will have unweighted in-degree at most $O(k)$. Therefore, when $m>nk$, we can contract all vertices with large unweighted in-degree to $s$, in order to decrease the edge size to $O(nk)$.
    \item At this point, we have obtained a graph where the original mincut (which was unbalanced) is a near-mincut of the new graph. Next, we create a (fractional) packing of edge-disjoint arborescences\footnote{An {\em arborescence} is a spanning tree in a directed graph where all the edges are directed away from the root.} in this graph using a multiplicative weights update framework (e.g.,~\cite{Young1995}). By duality, these arborescenes have the following property: if we sample $O(\log n)$ random arborescences from this packing, then there will be at least one arborescence whp such that the original mincut $1$-respects the arborescence. (A cut $1$-respects an arborescence if the latter contains just one edge from the cut.) 
    \item Thus, our task reduces to the following: given an arborescence, find the minimum weight cut in the original graph among all those that $1$-respect the arborescence. Our final technical contribution is to give an algorithm that solves this problem using $O(\log n)$ maxflow computations. For this purpose, we use a centroid-based recursive decomposition of the arborescence, where in each step, we use a set of maxflow calls that can be amortized on the original graph. The minimum cut returned by all these maxflow calls is eventually returned as the mincut of the graph.
\end{itemize}

We note that unlike both the Hao-Orlin algorithm and Gabow's algorithm that are both deterministic algorithms, our algorithm is randomized (Monte Carlo) and might yield the wrong answer with a small (inverse polynomial) probability. Derandomizing our algorithm, or matching our running time bound using a different deterministic algorithm, remains an interesting open problem.
}

\eat{

\paragraph{Previous Work.}\alert{more refs from Kent's manuscript}
The mincut problem has been studied in directed graphs over several decades. For unweighted graphs, Even and Tarjan~\cite{EvenT75} gave an algorithm for this problem that runs in $O(mn \cdot\min(\sqrt{m}, n^{2/3}))$ time. This was improved by Schnorr~\cite{Schnorr79} improved by this bound for certain graphs to $O(mn\lambda)$, where $\lambda$ is the value of the directed mincut. This was further improved by Mansour and Schieber~\cite{MansourS89} to $O(n\cdot \min(m, \lambda^2 n))$ after almost a decade of Schnorr's work. Mansour and Schieber's bound of $O(mn)$ was matched up to logarithmic factors for the more general case of weighted digraphs by Hao and Orlin~\cite{HaoO92}. Finally, Gabow~\cite{Gabow1995} gave an algorithm that runs in $\tO(m\lambda)$ which further refines this bound for graphs with small $\lambda$. These remained the fastest directed mincut algorithms for almost 30 years before our work.

\paragraph{Concurrent Work.}
Two recent results on algorithms for finding mincuts in directed graphs were obtained concurrently and independently of our work. First, Chekuri and Quanrud \cite{ChekuriQ21}  showed an exact algorithm with running time $\tO(n^2 U)$ if edge weights are integers between $1$ and $U$.
Second, Quanrud \cite{Quanrud21} has obtained an $(1+\epsilon)$-approximate algorithm that runs in $\tilde{O}(n^2/\epsilon^2)$ time.\footnote{Quanrud can also obtain $o(mn)$-time algorithms using the currently fastest maxflow algorithm on sparse graphs by Gao, Liu, and Peng \cite{GaoLP21}. We can also obtain the same time but for exact mincuts.} 
Both papers also extend their ideas to obtain approximation results for other problems as well, such as the vertex mincut problem.
}

\section{Minimum Cut Algorithms in Edge-weighted Directed Graphs}
\label{sec:exact-edge}

Given a directed graph $G = (V, E)$ with non-negative edge weights $w$ and a fixed root vertex $s$, we consider the problem of finding an $s$-mincut.
An \emph{$s$-arborescence} is a directed spanning tree rooted at $s$ such that all edges are directed away from $s$.
For simplicity, we assume that all edge weights $w$ are integers and are polynomially bounded.  We denote $\overline{U}=V\setminus U$.  Let $\outcut{U}$ be the set of edges from $U$ to $\overline{U}$, $\incut{U}$ be the set of edges from $\overline{U}$ to $U$, and let $\outcutsize(U)$ and $\incutsize(U)$ be the weight of the cut, i.e., $\outcutsize(U)=\sum_{e\in \outcut{U}}w(e), \incutsize{U}=\sum_{e\in \incut{U}}w(e)$.  Our goal is to compute the minimum cut $(\optSource, \optSink)$ where $\root \in \optSource = \arg\min_{s\in U \subset V} \outcutsize{U}$ and $\optSink = \overline{\optSource}$.
Let $F(m,n)$ denote the time complexity of $s$-$t$ maximum flow on a digraph with $n$ vertices and $m$ edges. The current record for this bound is $F(m, n) = \tO(m + n^{3/2})$~\cite{BrandLLSSSW21}. We emphasize that our directed mincut algorithm uses maxflow subroutines in a black box manner and therefore, any maxflow algorithm suffices. Correspondingly, we express our running times in terms of $F(m, n)$.


\begin{theorem}
  \label{thm:main-primitive}
  There is a Monte Carlo algorithm that finds a minimum $s$-cut whp in $\tO(\min\{mk, nk^2\}+F(m,n)\frac{n}{k})$ time, where $k$ is a parameter and $F(m,n)$ is the time complexity of $s$-$t$ maximum flow.
\end{theorem}
This section is devoted to prove \Cref{thm:main-primitive}.  If we set $k=m^{1/3}+n^{1/2}$ and use the $\tO(m+n^{3/2})$ max-flow algorithm, the time complexity becomes $\tO(nm^{2/3}+n^2)$, which establishes \Cref{thm:main}.  If we assume an hypothetical $\tO(m)$-time max-flow algorithm, then our result becomes $\tO(\min\{nm^{2/3},mn^{1/2}\})$ for $k=\min\{m^{1/3},n^{1/2}\}$.

We obtain \Cref{thm:main-primitive} via a new $s$-mincut algorithm. The algorithm considers the following two cases, computing a $s$-cut for each case and returning the minimum as its final output. The cases are split on $|\optSink|$ by a threshold $k>0$.


\begin{enumerate}
\item The first case aims to compute the correct mincut in the event that $|{\optSink}| > k$. In this case, if we randomly sample $\sink \in V$, then with probability at least $1/k$, $\sink \in \optSink$. Then $\optSink$ can be obtained via the maxflow from $s$ to $t$. Repeating the sampling $O(\frac{n}{k}\log n)$ times, we obtain the minimum $\root$-cut whp. The total running time for this case is $O(F(m,n)\frac{n}{k}\log n)$.
\eat{and is formalized in Lemma \ref{lem:size-balanced} below:
  \begin{lemma}
    \label{lem:size-balanced}
    Let $k \in \mathbb{N}$ be a given parameter.  If $|\optSink| > k$, then whp a mincut can be calculated in time $O(F(m,n) \cdot k \cdot \log n)$.
  \end{lemma}
  \begin{proof}
    Uniformly sample a list of $k = d \cdot (n/r) \cdot \lg n$ vertices $u_1,\ldots,u_k$, where $d$ is a large constant. Wlog, assume $|S^*| \le |\overline{S^*}|$, and let $\eta=\frac{|S^*|}{n}>\frac{r}{n}$.  With probability at least $1-2(1-\eta)^k\ge 1-2e^{-k\eta}\ge 1-2n^{-d}$, the list $u_1,\ldots,u_k$ contains at least one vertex from each of $S^*$ and $\overline{S^*}$. Hence, there exists $i$ such that $u_i$ and $u_{i+1}$ are on different sides of the $(S^*, \overline{S^*})$ cut.  By calculating maxflows for all ($u_i$, $u_{i+1}$) and ($u_{i+1}$, $u_i$) pairs, and reporting the smallest $s-t$ mincut in these calls, we return a global mincut whp.
  \end{proof}
  }
\item The second case is for the the event that $|{\optSink}| \le k$.
  Let $\lambda$ denote the value of the minimum rooted cut. By enumerating $O(\log n)$ powers of $2$, we can obtain an estimate $\varec$ such that $\ec \leq \varec \leq 2\ec$. For each value of $\varec$,  we apply \Cref{sparsification} to sparsify the graph in the following manner. First, \Cref{sparsification} returns a set of vertices $\sparseV \subseteq V$ such that $\root \in \sparseV$ and $\optSink \subseteq \sparseV$ whp. In particular, one can safely contract any vertex $v \in V \setminus \sparseV$ into $\root$ without affecting the minimum $\root$-cut. We contract $G$ accordingly and, overloading notation, let $G$ denote the contracted graph with vertex set $V_0$ henceforth. Second, \Cref{sparsification} returns a graph $\sparseG = (\sparseV, \sparseE)$ in which $\optSink$ still induces an $\epsmore$-approximate $s$-mincut, but the weight of the cut is now reduced to $\bigO{k \log{n}}$. We note that $\sparseG$ is not necessarily a subgraph of $G$. We then invoke  \Cref{thm:tree-packing} from \Cref{sec:arborescence} to fractionally pack an approximately maximum amount of $O(k \log n)$ $s$-arboresences in $\sparseG$ in $\apxO{m+\min{mk, nk^2}}$ time. In a random sample of $\bigO{\log n}$ $s$-arboresences from this packing, one of them will $1$-respect the $s$-mincut in $G$ (for appropriate $\varec$) whp:
  \begin{definition}
  \label{def:arborescence}
    A directed $s$-cut $(S, V \setminus S)$ \emph{$k$-respects} an $s$-arborescence if there are at most $k$ edges in the arborescence from $S$ to $V \setminus S$.
  \end{definition}
  Finally, for each of the $O(\log n)$ $s$-arborescences, the algorithm computes the minimum $\root$-cut that $1$-respects each arborescence. This algorithm is described in Algorithm~\ref{alg:tree} and proved in \Cref{thm:tree-alg} from \Cref{sec:maxflows}. It runs in $O((F(m,n) + m)\cdot \log n)$ time for each of the $O(\log n)$ arborescences.
\end{enumerate}
Combining both cases, the total running time becomes $\tilde O(\min\{mk, nk^2\} + F(m,n)\frac{n}{k})$, which establishes \Cref{thm:main-primitive}.

\begin{krq}
  \paragraph{Fast approximations.}
  The exact algorithm described above can be modified to produce a
  randomized $\apxO{n^2/\eps^2}$-time approximation algorithm that
  computes a $\epsmore$-approximate minimum $\root$-cut (hence also the
  global cut).  With logarithmic overhead, we can obtain a parameter
  $k$ such that $k / 2 \leq \sizeof{\optSink} \leq k$, where
  $\incut{\optSink}$ is the minimum  $\root$-cut. We then follow the
  same steps as in the exact algorithm, except whenever we compute the
  max-flow, we compute it in the sparsifier produced by
  \Cref{sparsification} instead. Since the sparsifier has at
  most $\apxO{nk / \eps^2}$ edges, we obtain a running time of the
  form
  \begin{math}
    \apxO{\min{nk^2 / \eps^2} + F(nk / \eps^2,n) (n/k)}.
  \end{math}
  For $F(m,n) = \apxO{m + n^{1.5}}$, this gives a running time of $\apxO{n^2 / \eps^2}$.

  \begin{theorem}
    For $\eps \in (0,1)$, an $\epsmore$-approximate minimum $s$-cut  (hence global minimum cut) can be computed in $\tO(n^2 / \eps^2)$ time whp.
  \end{theorem}

  We remark that $\apxO{n^2 / \eps^2}$ can also be obtained by using a local connectivity algorithm similar to the approach for vertex mincuts in \refsection{apx-vc}, instead of via an arboresence packing. See \cite{Quanrud21} for details.
\end{krq}

\paragraph{Organization of the rest of this section.}
The following subsections present each step of the algorithm described above. First, we establish the partial sparsification subroutine in \Cref{sec:sparsification}. Next, in \Cref{sec:arborescence}, we obtain an arborescence packing from which sampling yields an arborescence that is $1$-respected by the mincut. Finally, in \Cref{sec:maxflows}, we describe the algorithm to retrieve the mincut among those that $1$-respect a given arborescence.

\subsection{Partial Sparsification}
\label{sec:sparsification}


\labelsection{rec-sparsification}

\renewcommand{\refstep}[1]{\ref{step:#1}}%

This section aims to reduce mincut value to $\tO(k)$ and
edge size to $\min\{m,O(nk\log{n}/\epsilon^2\})$ while keeping $\outcut{\optSource}$ a
$(1+\epsilon)$-approximate $s$-mincut for a constant $\epsilon > 0$ that
we will fix later.  Our algorithm in this stage has three
steps. First, we use random sampling to discretize and scale down the
expected value of all cuts such that the expected value of the mincut
$\outcutsize(\optSource)$ becomes $\tO(k)$. We also claim that $\outcut(\optSource)$
remains an approximate mincut {\em among all unbalanced cuts} by using
standard concentration inequalities. However, since the number of
balanced cuts far exceeds that of unbalanced cuts, it might be the
case that some balanced cut has now become much smaller in weight than
all the unbalanced cuts. This would violate the requirement that
$\outcut(\optSource)$ should be an approximate mincut in this new
graph. This is where we need our second step, where we overlay a star
on the sampled graph to raise the values of all balanced $s$-cuts
above the expected value of $\outcut(\optSource)$ while only increasing the
value of $\outcut(\optSource)$ by a small factor. The third step leverages
the fact that, after scaling, the minimum edge weight is $1$, and the
minimum cut is $\bigO{k \log{n} / \eps^2}$. It follows that any vertex
with in-degree at least a constant factor greater than
$k \log{n} / \eps^2$ cannot be in the sink component, and can be
safely contracted into the root without affecting the $s$-mincut.

The first two steps described above are implemented in the next lemma,
whose proof appears in \Cref{sec:proof skeleton}.


\begin{lemma}
  \labellemma{apx-rec-skeleton} Let $G = (V,E)$ be a directed graph
  with positive edge weights. Let $\root \in V$ be a fixed root
  vertex. Let $\eps \in (0,1)$, let $\ec > 0$, and let
  $k \in \naturalnumbers$ be given parameters. Suppose there is an
   $\root$-mincut of the form $\incut{\optSink}$, where $\root \not\in T^*$,
  $\ec/2 \leq \incutsize{\optSink}{G} \leq 2 \ec$ and
  $\sizeof{\optSink} \leq k$.  In randomized nearly linear time, one can
  compute a randomized directed and reweighted subgraph
  $\sparseG = (V, \sparseE)$, where $\sparseV \subseteq V$ and
  $\root \in \sparseV$ with the following properties.
  \begin{properties}
  \item \label{rec-skeleton-small-cut} $\sparseG$ has integral edge weights and the minimum $\root$-cut has weight at most
    $\bigO{k \log{n} / \eps^2}$.
  \item \label{rec-skeleton-min-cut} $\incut{\optSink}{G_0}$ is a
    $\epsmore$-approximate minimum $\root$-cut in $G_0$.
  \item \label{rec-skeleton-apx-cuts} Every $\alpha$-approximate
    minimum $\root$-cut in $G_0$ induces an $\epsmore \alpha$-approximate
    minimum $\root$-cut in $G$.
  \end{properties}
\end{lemma}

The preceding lemma importantly allows us to reduce the weight of the minimum cut to roughly $k$, where $k$ is the number of the vertices in the sink component. However it has not actually reduced the size in the graph, in terms of the number of edges. This is accomplished by our third step, which we formalize in  the following lemma that reduces the number of edges to $\tO(nk)$. 

\begin{lemma}
  \label{thm:sparsification}    %
  \label{sparsification}        %
  \labellemma{apx-rec-sparsification} Let $G = (V,E)$ be a directed
  graph with positive edge weights.  Let $\root \in V$ be a fixed root
  vertex. Let $\eps \in (0,1)$, $\ec > 0$, and
  $k \in \naturalnumbers$ be given parameters.  Suppose there is a
  minimum $\root$-cut is of the form $\incut{\optSink}$, where
  $\ec/2 \leq \incutsize{\optSink}{G} \leq 2 \ec$ and $\sizeof{\optSink} \leq k$.

  In randomized nearly linear time, one can compute a randomized directed and edge-weighted graph $\sparseG = (\sparseV, \sparseE)$, where $\sparseV \subseteq V$ and $\root \in \sparseV$.
  \begin{properties}
  \item $\sparseG$ has integral edge weights and the $s$-mincut in $\sparseG$ has weight at most $\bigO{k \log{n} / \eps^2}$.
  \item $\sparseG$ has at most
    $\sizeof{\sparseE} = \min{m, \bigO{nk \log{n} / \eps^2}}$ edges.
  \item $\sparseG$ is a subgraph of the graph obtained by contracting
    $V \setminus V_0$ into $\root$ in $G$.
  \item We have $\optSink \subseteq V_0$, and $\incut{\optSink}{G_0}$
    is an $\epsmore$-approximate minimum $\root$-cut in $G_0$.
  \item Every $\epsmore$-approximate minimum $\root$-cut in $G_0$ induces
    an $\epsmore^2$-approximate minimum $\root$-cut in $G$.
  \end{properties}
\end{lemma}

\begin{proof}
  Consider the reweighted subgraph produced by
  \reflemma{apx-rec-skeleton}, which we denote by $G_1 = (V, E_1)$. We
  claim that every vertex $v \in \optSink$ has \emph{unweighted}
  in-degree at most $\bigO{k \log{n} / \eps^2}$. Indeed, at most $k-1$
  of these edges are from other vertices in $T$, and the remaining
  edges must be in $\incut{\optSink}{G_1}$. But
  $\incut{\optSink}{G_1}$ has at most
  $\incutsize{\optSink}{G_1} = \bigO{k \log{n} / \eps^2}$ edges by
  properties \ref{rec-skeleton-small-cut} and
  \ref{rec-skeleton-min-cut} of \reflemma{apx-rec-skeleton}.

  Let $G_0 = (V_0, E_0)$ be the graph obtained from $G_1$ by
  contracting all vertices with unweighted in-degree
  $\geq c k \log{n} / \eps^2$ for a sufficiently large constant $c$
  that excludes all vertices in $\optSink$.  It is easy to see that
  $G_0$ satisfies the claimed properties, particularly as contractions
  into $\root$ do not decrease the $s$-mincuts, and
  $\incut{\optSink}{G_0}$ is preserved exactly as in $G_1$.
\end{proof}

\subsection{Finding a 1-respecting Arborescence}
\label{sec:arborescence}
In this section, we assume that there is an unbalanced $\root$-mincut and show how to obtain an $\root$-arborescence that 1-respects the $\root$-mincut.  More formally, we prove the following:

\begin{lemma}
  \label{thm:tree-packing}
  Given weighted digraph $G$ and a fixed root vertex $\root$, suppose the sink side of an $\root$-mincut $\optSink$ has at most $k$ vertices. In $O(m \log n +\min\{mk\log^2 n,nk^2\log^3 n\})$ time, we can find $O(\log n)$ $\root$-arborescences on vertex set $V_0\supset \optSink$, such that whp an $s$-mincut 1-respects at least one of them.
\end{lemma}

The idea of this lemma is as follows. First, we apply \Cref{sparsification} to our graph $G$ and obtain the graph $G_0$. Whp, a minimum $s$-cut $\incut{\optSink}$ in $G$ corresponds to a $(1+\epsilon)$-approximate minimum $\root$-cut in $G_0$. It remains to find an arborescence in $G_0$ that 1-respects $\incut{\optSink}$. To do this, we employ a multiplicative weight update (MWU) framework. The algorithm begins by setting all edge weights to be uniform (say, weight $1$). Then, we repeat the following for $O(k\log(n)/\epsilon^{2})$ rounds: in each round, we find a minimum weight arborescence in $\tO(m)$ time and multiplicatively increase the weight of every edge in the arborescence.
Using the fact that there is no duality gap between arborescence packing and mincut \cite{Edmonds73,Gabow1995} (see \Cref{lem:duality}), a standard MWU analysis implies that these arborescences that we found, after scaling, form a $(1+\epsilon)$-approximately optimal fractional arborescence packing. So our arborescence crosses $\optSink$ at most $(1+O(\epsilon))<2$ times on average. Thus, if we sample $O(\log n)$ arborescences from this set, one of them will 1-respect $\optSink$ whp.\footnote{This should be compared with Karger's mincut algorithm in the undirected case, where there is a factor $2$ gap, and hence Karger can only guarantee a $2$-respecting tree in the undirected case.} 
To obtain the running time bound, we note that each iteration of the MWU framework requires us to find the minimum cost $s$-arborescence, for which an $\tO(m)$-time algorithm is known~\cite{Gabow1986}.

Since the argument above is a standard application of the MWU framework, we give the detailed proof in \Cref{sec:packarb_details}.

\subsection{Mincut Given 1-respecting Arborescence}
\label{sec:maxflows}

We propose an algorithm (Algorithm~\ref{alg:tree}) that uses $O(\log n)$ maxflow subroutines to find the minimum $s$-cut  that $1$-respects a given $s$-arborescence.
The result is formally stated in Theorem \ref{thm:tree-alg}.

\begin{theorem} \label{thm:tree-alg}
Consider a directed graph $G = (V, E)$ with polynomially bounded edge weights $w_e > 0$.
Let $s\in V$ be a fixed root vertex and $S\ni s$ be the source side of a fixed $s$-mincut.
Given an $s$-arborescence $T$ with $|T \cap \outcut{S}| = 1$, Algorithm~\ref{alg:tree} outputs a $s$-mincut of $G$ in time $O((F(m,n) + m)\cdot \log n)$.
\end{theorem}

We first give some intuition for Algorithm~\ref{alg:tree}.
Because $s \in S$, if we could find a vertex $t \in \overline{S}$, then computing the $s$-$t$ mincut using one maxflow call would yield a global mincut of $G$.  However, we cannot afford to run one maxflow between $s$ and every other vertex in $G$.
Instead, we carefully partition the vertices into $\ell = O(\log n)$ sets $(C_i)_{i=1}^{\ell}$.
We show that for each $C_i$, we can modify the graph appropriately so that it allows us to (roughly speaking) compute the maximum flow between $s$ and every vertex $c \in C_i$ using one maxflow call.

More specifically, Algorithm~\ref{alg:tree} has two stages. In the first stage, we compute a \emph{centroid decomposition} of $T$. Recall that a centroid of $T$ is a vertex whose removal disconnects $T$ into subtrees with at most $n/2$ vertices. This process is done recursively, starting with the root $s$ of $T$. We let $P_1$ denote the subtrees resulting from the removal of $s$ from $T$. In each subsequent step $i$, we compute the set $C_i$ of the centroids of the subtrees in $P_i$. We then remove the centroids and add the resulting subtrees to $P_{i+1}$. This process continues until no vertices remain.

In the second stage, for each layer $i$, we construct a directed graph $G_i$ and perform one maxflow computation on $G_i$. The maxflow computation on $G_i$ would yield candidate cuts for every vertex in $C_i$, and after computing the appropriate maximum flow across every layer, we output the minimum candidate cut as the minimum cut of $G$. The details are presented in Algorithm~\ref{alg:tree}.

\begin{algorithm}[t]
\caption{Finding an $s$-mincut.}
\label{alg:tree}
\SetKwInOut{Input}{Input}
\SetKwInOut{Output}{Output}
\setcounter{AlgoLine}{0}
\Input{An arborescence $T$ rooted at $s\in S$ such that $S$ 1-respects $T$.}
\smallskip
// Stage I: Build centroid decomposition. \\
\smallskip
Let $C_0 = \{s\}$, $P_1 =$ the set of subtrees obtained by removing $s$ from $T$, and $i = 1$. \\
\While{$P_i\ne\varnothing$}{
Initialize $C_i$ (the centroids of $P_i$) and $P_{i+1}$ as empty sets. \\
\For{each subtree $U\in P_i$}{
Compute the centroid $u$ of $U$ and add it to $C_i$. \\
    Add all subtrees generated by removing $u$ from $U$ to $P_{i+1}$.
    }
Set $\ell = i$ and iterate $i = i+1$.
}
\smallskip
// Stage II: Calculate integrated maximum flow for each layer. \\
\smallskip
\For{$i = 1$ {\bf to} $\ell$}{
Construct a digraph $G_i = (V \cup \{t_i\}, E_1 \cup E_2 \cup E_3)$ as follows (see \Cref{fig:my_label}): \label{alg:step11}\\
    \quad 1) Add edges $E_1 = E \cap \cup_{U \in P_i} (U\times U)$ with capacity equal to their original weight.\label{alg:step12}\\
    \quad 2) Add edges $E_2 = \{(s,v):(u,v)\in E\setminus E_1\}$ with capacity of $(s,v)$ equal to the original weight of $(u,v)$. \\ 
    \quad 3) Add edges $E_3 = \{(u,t_i):u\in C_i\}$ with infinite capacity.\label{alg:step14}\\
Compute the maximum $s$-$t_i$ flow $f^*_i$ in $G_i$. \label{alg:step15} \\
For each component $U\in P_i$ with centroid $u$, the value of $f^*_i$ on edge $(u,t_i)$ is a candidate cut value, and the nodes in $U$ that can reach $u$ in the residue graph is a candidate for $\overline{S}$.
}
Return the smallest candidate cut value and the corresponding $(S, \overline{S})$ as an $s$-mincut.
\end{algorithm}

\begin{figure}
    \centering
    \includegraphics[width=0.75\linewidth]{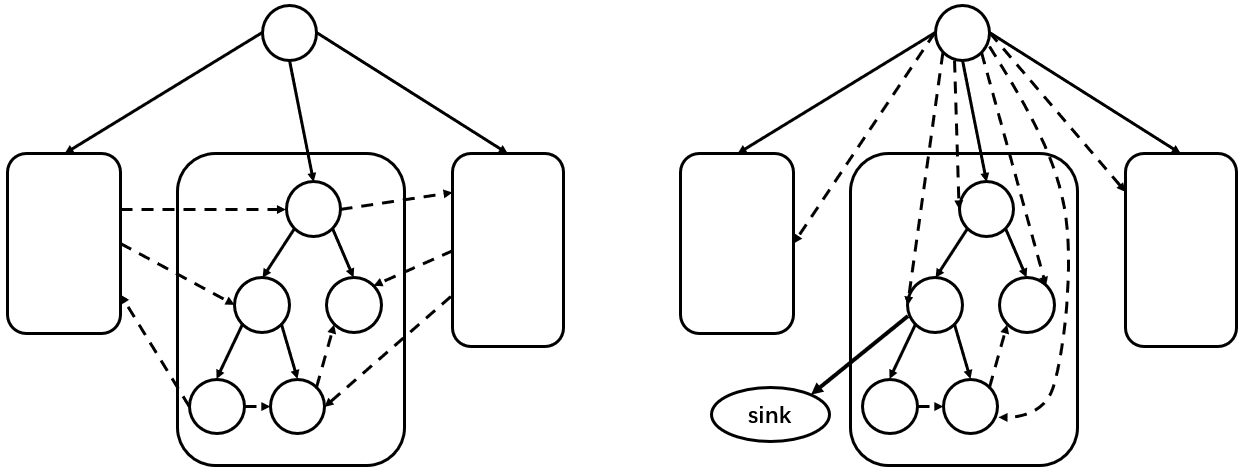}
    \caption{
    Construction of auxiliary graph $G_i$ in Algorithm \ref{alg:tree}.
    Solid lines represent the arborescence $T$.
    Dashed lines are other edges in the graph.
    Rectangles are sets formed by the first level of centroid decomposition.
    Left: The original graph.
    Right: The part of $G_1$ solving the case that the mincut separates root and the centroid of the middle subtree.
    }
    \label{fig:my_label}
\end{figure}

We first state two technical lemmas that we will use to prove Theorem~\ref{thm:tree-alg}.


\begin{lemma} \label{lem:tree-barS-inone}
Recall that $P_i$ is the set of subtrees in layer $i$ and $C_i$ contains the centroid of each subtree in $P_i$.
If $C_{j} \subseteq S$ for every $0 \le j < i$, then $\overline{S}$ is contained in exactly one subtree in $P_i$, and consequently, at most one vertex $u \in C_i$ can be in $\overline{S}$.
\end{lemma}

\begin{lemma} \label{lem:tree-fi}
Let $G_i$ be the graph constructed in Step~\ref{alg:step11} of Algorithm~\ref{alg:tree}.
Let $f^*_i$ be a maximum $s$-$t_i$ flow on $G_i$ as in Step~\ref{alg:step15}.
For any $U \in P_i$ with centroid $u$, the amount of flow $f^*_i$ puts on edge $(u, t_i)$ is equal to the value of the minimum cut from $\overline{U}$ to $u$.
\end{lemma}

We defer the proofs of Lemmas~\ref{lem:tree-barS-inone}~and~\ref{lem:tree-fi}, and first use them to prove Theorem~\ref{thm:tree-alg}.

\begin{proof}[Proof of Theorem~\ref{thm:tree-alg}]
We first prove the correctness of Algorithm~\ref{alg:tree}.

Because $C_0 = \{s\}$ and $s \in S$, and the $C_i$'s form a disjoint partition of $V$, there must be a layer $i$ such that for the first time, we have a centroid $u \in C_i$ that belongs to $\overline{S}$.
By Lemma~\ref{lem:tree-barS-inone}, we know that $\overline{S}$ must be contained in exactly one subtree $U \in P_i$, and hence $u$ must be the centroid of $U$.
In summary, we have $u \in \overline{S}$ and $\overline{S} \subseteq U$.

Consider the graph $G_i$ constructed for layer $i$.
By Lemma~\ref{lem:tree-fi}, based on the flow $f^*_i$ puts on the edge $(u, t_i)$, we can recover the value of the minimum cut from $\overline{U}$ to $u$.
Because $\overline{S} \subseteq U$ (or equivalently $\overline{U} \subseteq S$) and $u \in \overline{S}$, the cut $(S, \overline{S})$ is one possible cut that separates $\overline{U}$ and $u$. Therefore, the flow that $f^*_i$ puts on the edge $(u, t_i)$ is equal to the $s$-mincut value in $G$.

In addition, the candidate cut value for any other centroid $u'$ of a subtree $U' \in P_i$ must be at least the mincut value between $s$ and $u'$.
This is because the additional restriction that the cut has to separate $\overline{U'}$ from $u'$ can only make the mincut value larger, and the value of this cut in $G_i$ is equal to the value of the same cut in $G$.
Therefore, the minimum candidate cut value in all $\ell$ layers must be equal to the $s$-mincut value of $G$.

Now we analyze the running time of Algorithm~\ref{alg:tree}. We can find the centroid of an $n$-node tree in time $O(n)$ (see e.g., \cite{Megiddo1981}).
The total number of layers $\ell = O(\log n)$ because removing the centroids reduces the size of the subtrees by at least a factor of $2$. Thus, the running time of Stage I of Algorithm~\ref{alg:tree} is $O(n \log n)$. In Stage II, we can construct each $G_i$ in $O(m)$ time and every $G_i$ has $O(m)$ edges.
Since there are $O(\log n)$ layers and the maximum flow computations take a total of $O(MF(m,n)\cdot \log n)$ time,
the overall runtime is $O(n \log n + (MF(m, n) + m) \log n) = O((MF(m, n) + m) \log n)$.
\end{proof}

Before proving Lemmas~\ref{lem:tree-barS-inone}~and~\ref{lem:tree-fi} we first prove the following lemma.

\begin{lemma} \label{lem:tree-barS-path}
If $x$ and $y$ are vertices in $\overline{S}$, then every vertex on the (undirected) path from $x$ to $y$ in the arborescence $T$ also belongs to $\overline{S}$.
\end{lemma}

\begin{proof}
Consider the lowest common ancestor $z$ of $x$ and $y$.
Because there is a directed path from $z$ to $x$ and a directed path from $z$ to $y$, we must have $z \in \overline{S}$.
Otherwise, there are at least two edges in $T$ that go from $S$ to $\overline{S}$.

Because $s \in S$ and $z \in \overline{S}$, there is already an edge in $T$ (on the path from $s$ to $z$) that goes from $S$ to $\overline{S}$.
Consequently, all other edges in $T$ cannot go from $S$ to $\overline{S}$, which means the entire path from $z$ to $x$ (and similarly $z$ to $y$) must be in $\overline{S}$.
\end{proof}

Recall that Lemma~\ref{lem:tree-barS-inone} states that if all the centroids in previous layers are in $S$, then $\overline{S}$ is contained in exactly one subtree $U$ in the current layer $i$.

\begin{proof}[Proof of Lemma~\ref{lem:tree-barS-inone}]
For contradiction, suppose that there exist distinct subtrees $U_1$ and $U_2$ in $P_i$ and vertices $x,y \in \overline{S}$ such that $x \in U_1$ and $y \in U_2$.

By Lemma~\ref{lem:tree-barS-path}, any vertex on the (undirected) path from $x$ to $y$ also belongs to $\overline{S}$.
Consider the first time that $x$ and $y$ are separated into different subtrees.
This must have happened because some vertex on the path from $x$ to $y$ is removed. However, the set of vertices removed at this point of the algorithm is precisely $\bigcup_{0 \le j < i} C_j$, but our hypothesis assumes that none of them are in $\overline{S}$.
This leads to a contradiction and therefore $\overline{S}$ is contained in exactly one subtree of $P_i$.

It follows immediately that at most one centroid $u \in C_i$ can be in $\overline{S}$.
\end{proof}

Next we prove Lemma~\ref{lem:tree-fi}, which states that the maximum flow between $s$ and $t_i$ in the modified graph $G_i$ allows one to simultaneously compute a candidate mincut value for each vertex $u \in C_i$.

\begin{proof}[Proof of Lemma~\ref{lem:tree-fi}]
First observe that the maxflow computation from $s$ to $t_i$ in $G_i$ can be viewed as multiple independent maxflow computations.
The reason is that, for any two subtrees $U_1, U_2 \in P_i$, there are only edges that go from $s$ into $U_1$ and from $U_1$ to $t_i$ in $G_i$ (similarly for $U_2$), but there are no edges that go between $U_1$ and $U_2$.

The above observation allows us to focus on one subtree $U \in P_i$.
Consider the procedure that we produce $G_i$ from $G$ in Steps~\ref{alg:step12}~to~\ref{alg:step14} of Algorithm~\ref{alg:tree}.
The edges with both ends in $U$ are intact (the edge set $E_1$).
If we contract all vertices out of $U$ into $s$, then all edges that enter $U$ would start from $s$, which is precisely the effect of removing cross-subtree edges and adding the edges in $E_2$.
One final infinity-capacity edge $(u, t_i) \in E_3$ connects the centroid of $U$ to the super sink $t_i$.

Therefore, the maximum $s$-$t_i$ flow $f^*_i$ computes the maximum flow between $\overline{U}$ and $u \in U$ simultaneously for all $U \in P_i$, whose value is reflected on the edge $(u, t_i)$.
It follows from the maxflow mincut theorem that the flow on edge $(u, t_i)$ is equal to the mincut value between $\overline{U}$ and $u$ in $G$ (i.e., the minimum value $w(A, \overline{A})$ among all $A \subset V$ with $\overline{U} \subseteq A$ and $u \in \overline{A}$).
\end{proof}


\begin{krq}

\section{Minimum Cut Algorithms in Vertex-weighted Directed Graphs}

\labelsection{apx-vc}

In this section we present the approximation algorithm for the minimum
rooted and global vertex cut. Similar to \Cref{sec:exact-edge}, the
main focus is on rooted cuts, and the algorithm is presented in three
main parts. All three parts are parameterized by values $\vc > 0$ and
$k \in \naturalnumbers$ that, in principle, are meant to be constant
factor estimates for the weight and the number of
vertices in the sink component of the minimum rooted vertex cut. The
first part, in \Cref{rvc-sparsification}, presents the
sparsification lemma that reduces the number of edges to roughly $nk$
and the rooted mincut to roughly $k$ in a graph with integer
weights. This sparsifier is used in the remaining two parts. The
second part, in \Cref{local-rvc}, gives a roughly $nk^2$ time
approximation algorithm for the minimum rooted cut via a new local
flow algorithm. The third part, in \Cref{sample-rvc}, gives a roughly
$n^2 + n^{2.5} / k$ time approximation algorithm via sampling and
$(\root,\sink)$-flow (as with minimum edge cuts before). Finally, in
\Cref{rvc-all-together}, we balance terms to obtain the claimed
running time for rooted cut. The rooted vertex mincut
algorithm then leads to a global vertex mincut algorithm via an
argument due to \cite{hrg-00} (with some modifications).

\subsection{Partial Sparsification}\label{rvc-sparsification}
The first part is a sparsification lemma that preserves rooted vertex
cuts where the number of vertices in the sink component is below some
given parameter.  It is similar in spirit to
\Cref{thm:sparsification}, but with some necessary changes as we are
now preserving the vertex mincut rather than edge
mincut.   We give a brief overview of the algorithm, highlighting in particular the differences from the partial edge cut sparsifier. The proof and algorithmic details are deferred to \Cref{proof-apx-rvc-sparsification}.

At a high level, the following sparsifier for vertex cuts randomly samples the \emph{vertex} weights so that the weights are integral, and the weight of the minimum vertex cut becomes $\bigO{k \log{n} / \eps^2}$. Similar to the partial edge sparsifier, this rounding is calibrated to preserve $\root$-cuts with (roughly) $k$ or fewer vertices in the sink components. To pad the weight of vertex $\root$-cuts with large sink components, we add an weighted auxiliary vertex on a short directed path between $\root$ and each vertex (as opposed to just adding an edge from $\root$, as we did for edge cuts). If a sampled weight of a vertex $v$ is $0$, we cannot simply drop the vertex from the graph (in the way we can drop weight $0$ edges) since the vertex may be in the sink component of the min $r$-cut. Instead we remove all outgoing edges from $v$.
Also, when we detect that a vertex $v$ cannot be in the sink component (by a similar counting argument as before), rather than contract $v$ into $\root$ (which may effect the min vertex $\root$-cut), we replace all of the incoming edges to $v$ with a single edge from $\root$. The culmination of these modifications is a similar net effect as for edge cuts: a graph with $\bigO{n k \log{n} / \eps^2}$ that preserves the sink component of the minimum vertex $\root$-cut.
That said, the following bounds are more detailed than the bounds for preserving the edge cut in \Cref{thm:sparsification}. These additional properties play a critical role in the customized local flow algorithms presented later.

In the following, let $\outneighbors{v}{G}$ denote
the set of out-neighbors of $v$ in the graph $G$. We omit $G$ and
simply write $\outneighbors{v}$ when $G$ can be inferred from the
context.
\begin{lemma}
  \labellemma{apx-rvc-sparsification} Let $G = (V,E)$ be a directed
  graph with positive vertex weights.  Let $\root \in V$ be a fixed
  vertex.  Let $k, \vc > 0$ be given parameters.  Let
  $V' = V \setminus \parof{\setof{\root} \cup \outneighbors{\root}}$.
  In randomized linear time, one can compute a randomized directed and
  vertex-weighted graph $\sparseG = (\sparseV, \sparseE)$, and a
  scaling factor $\tau > 0$, with the following properties.
  \begin{properties}
  \item \label{rvcs-root} $\root \in \sparseV$.
  \item \label{rvcs-sinks} Let
    $\sparseV' = \sparseV \setminus \parof{\setof{\root} \cup
      \outneighbors{\root}{\sparseG}}$. We have $\sparseV' = V'$.
  \item \label{rvcs-weights} $\sparseG$ has integer vertex weights
    between $0$ and
    $\bigO{k \log{n} / \eps^2}$.
  \item \label{rvcs-degree} Every vertex $v \in \sparseV$ has at most
    $\bigO{k \log{n} / \eps^2}$ incoming edges.
  \item \label{rvcs-zero} Every vertex $v$ with weight $0$ has no
    outgoing edges.
  \item \label{rvcs-cut-weights} With high probability, for all
    $S \subseteq V'$, the weight of the vertex in-cut induced by $S$
    in $\sparseG$ (up to scaling by $\tau$) is at least the minimum of
    the $\epsless$ times the weight of the induced vertex in-cut in
    $G$ or $c \vc$ (for any desired constant $c > 1$), and at most
    $\epsmore$ times its weight in $G$ plus $\eps \vc \sizeof{S} / k$.
  \item \label{rvcs-small-cut-sets} With high probability, for all
    $S \subseteq V'$ such that $\sizeof{S} \leq k$ and the weight of
    the induced vertex in-cut is $\leq \bigO{\vc}$, we have
    $S \subseteq \sparseV'$. (That is, $S$ is still the sink component
    of an $\root$-cut in $\sparseG$.)
  \end{properties}
  In particular, if the minimum vertex $\root$-cut has weight
  $\bigTheta{\vc}$, and the sink component of a minimum vertex $\root$-cut
  has at most $k$ vertices, then with high probability $\sparseG$
  preserves the minimum vertex $\root$-cut up to a
  $\apxmore$-multiplicative factor.
\end{lemma}

As stated above, the proof is deferred to \Cref{proof-apx-rvc-sparsification}.



\subsection{Rooted vertex mincut for small sink components}

\label{local-rvc}

This section presents an approximation algorithm for rooted vertex
mincut for the particular setting where the sink component is
small. In particular, we are given an upper bound $k$ on the number of
vertices in the sink component, and want to obtain running times of
the form $n \poly{k}$.  When a similar situation arose previously for
small integer capacities in \cite{cq-simple-connectivity},
\cite{cq-simple-connectivity} modified a local algorithm from
\cite{forster+20} which works well for unweighted graphs. Here,
while \reflemma{apx-rvc-sparsification} produces relatively sparse
graphs with integral vertex capacities, the vertex capacities imply
imply that the
algorithm from \cite{cq-simple-connectivity,forster+20} would take
roughly $nk^3 / \eps^5$ time.  This section develops an alternative
algorithm that is inspired by these local algorithms for (global and
rooted) vertex cuts, but reduces the dependency on $k$ to $k^2$.
Compared to \cite{forster+20,cq-simple-connectivity}, the algorithm
here is designed to take full advantage of the properties of the
graph produced by \reflemma{apx-rvc-sparsification}.  These
modifications have some tangible benefits.  First, it improves the
dependency on $k$ and $\epsilon$.  Second, the local subroutine here
is deterministic whereas before they were randomized.  Third and last,
as suggested by the better running time and the determinism, the
version presented here is arguably simpler and more direct than the
previous algorithms (for this setting).


\begin{lemma}
  \labellemma{local-apx-rvc} Let $G = (V,E)$ be a directed graph with
  positive vertex weights.  Let $r \in V$ be a fixed root vertex.  Let
  $\eps \in (0,1)$, $\vc > 0$ and $k \in \naturalnumbers$ be given
  parameters.  There is a randomized linear time Monte Carlo algorithm
  that, with high probability, produces a deterministic data structure
  that supports the following query.

  For
  $t \in V' \defeq V \setminus \parof{\setof{\root} \cup
    \outneighbors{\root}}$, let $\vc_{t,k}$ denote the weight of the
  minimum $(\root,t)$-vertex cut such that the sink component has at most
  $k$ vertices. Given $t \in V'$, deterministically in
  $\bigO{k^3 \log{n}^2 / \eps^4}$ time, the data structure either (a)
  returns the sink component of a minimum $(\root,t)$-vertex cut of weight
  at most $(1 + \eps) \vc_{t,k}$, or (b) declares that
  $\vc_{t,k} > \vc$.
\end{lemma}

\begin{proof}
   Given $\root$, $\vc$, $k$, and $\eps$, let $\eps' = c \eps$ for a sufficiently small constant $c > 0$. We first apply \reflemma{apx-rvc-sparsification} to $G$ with root $\root$ and parameters $\vc$, $k$, and $\eps'$.  This produces a vertex capacitated graph  $\sparseG= (V_0,E_0)$ with $V \subset V_0$.  We highlight the
  features that we leverage.  All new vertices (in $V_0 \setminus V$)
  are in $\outneighbors{\root}{\sparseG}$; that is, $V'$ equals
  \begin{math}
    \sparseV' \defeq \sparseV \setminus \parof{\setof{\root} \cup
      \outneighbors{\root}{\sparseG}}.
  \end{math}
  Put alternatively, none of the new vertices is in the sink
  component of any $\root$-cut.  The vertex weights are integers between
  $0$ and $\bigO{k \log{n} / \eps^2}$.  Every vertex has unweighted
  in-degree at most $\bigO{k \log{n} / \eps^2}$.  Every vertex with
  weight $0$ has no outgoing edges.

  With high probability, we have the following guarantees on the
  vertex $\root$-cuts of $\sparseG$.  The vertex weights in $\sparseG$ are
  scaled so that a weight of $\vc$ in $G$ corresponds to weight
  $\bigO{k \log{n}/ \eps^2}$ in $\sparseG$.  Modulo scaling, every
  vertex $\root$-cut in $\sparseG$ has weight no less than the minimum of
  its weight in $G$ and $2 \vc$.  Additionally, modulo scaling, for
  every vertex $\root$-cut in $G$ with capacity at most $\vc$ and at most
  $k$ vertices in the sink component, the corresponding vertex cut in
  $\sparseG$ has weight at most a $c_0 \eps \vc$ additive factor
  larger than in $G$, for any desired constant $c_0 > 0$.  We consider
  the algorithm to fail if the cuts are not preserved in the sense
  described above.

  Given $t \in V$, the data structure will search for a small
  $(\root,t)$-cut in $\sparseG$ via a customized, edge-capacitated flow
  algorithm. This algorithm may or may not return the sink component
  of $(\root,t)$-cut. If the search does return a sink component, and the
  corresponding vertex in-cut in $\sparseG$ has weight that, upon
  rescaling back to the scale of the input graph $G$, is at most
  $(1 + \eps / 2) \vc$, the data structure returns it. Otherwise the
  data structure indicates that $\vc_{t,k} > \vc$.

  Proceeding with the flow algorithm, let $\revG$ be the reverse of
  $\sparseG$, and let $\splitG$ be the standard ``split-graph'' of
  $\revG$ modeling vertex capacities with edge capacities.  We recall
  that the split graph splits each vertex $v$ into an auxiliary
  ``in-vertex'' $\vin$ and an auxiliary ``out-vertex'' $\vout$.  For
  each $v$ there is a new edge $(\vin,\vout)$ with capacity equal to
  the vertex capacity of $v$.  Each edge $(u,v)$ is replaced with an
  edge $(\uout,\vin)$ with capacity\footnote{Usually, this edge is set
    to capacity $\infty$, but either the weight of $u$ or the weight
    of $v$ are also valid.}  equal to the vertex capacity of $u$.
  Every $(\root,t)$-vertex cut in $\sparseG$ maps to a
  $(\tout,\rin)$-edge cut in $\revG$ with the same capacity.  Any
  $(\tout,\rin)$-edge capacitated cut maps to a $(\root,t)$-vertex cut
  in $\sparseG$ (with negligible overhead in the running time). Now,
  recall that for each $v \in V'$, the sparsification procedures
  introduces an auxiliary path $(\root,a_v,\rin)$ where $a_v$ is was
  given weight $\bigTheta{\eps \vc / k}$. It is convenient to replace
  the corresponding auxiliary path $(\vout, a_v^-, a_v^+, \rin)$ in
  $\revG$ with a single edge $(\vout,\rin)$ with capacity equal to the
  weight of $a_v$.  This does not effective the weight of the minimum
  $(\tout,\rin)$-edge cut for any $t \in V'$.  This adjustment can be
  easily made within the allotted preprocessing time.

  In this graph, given $t \in V'$, we run a specialization of the
  Ford-Fulkerson algorithm \cite{ford-fulkerson} that either computes
  a minimum $(\tout,\rin)$-cut or concludes that the minimum
  $(\tout,\rin)$-cut is at least $\bigO{k \log{n} / \eps^2}$ (which
  corresponds to $\bigO{\vc}$ in $G$) after
  $\bigO{k \log{n} / \eps^2}$ iterations.  To briefly review, each
  iteration in the Ford-Fulkerson algorithm searches for a path from
  $t$ to $\root$ in the residual graph. If such a path is found, then it
  routes one unit of flow along this path, and updates the residual
  graph by reversing (one unit capacity) of each edge along the path.
  After $\ell$ successful iterations we have a flow of size $\ell$.
  If, after $\ell$ iterations, there is no path in the residual graph
  from $\tout$ to $\rin$, then the set of vertices reachable from $t$
  gives a minimum $(\tout,\rin)$-cut of size $\ell$.  Observe that
  updating the residual graph along a $(\tout,\rin)$-path preserves
  the weighted in-degree and out-degree of every vertex except $\tout$
  and $\rin$. The weighted out-degree of $\tout$ decreases by 1 and
  the weighted in-degree of $\rin$ changes by $1$. Moreover, updating
  the residual graph along a path increases the unweighted out-degree
  of any vertex by at most one, since a path contains at most one edge
  going into any single vertex.  Since every vertex initially has
  unweighted out-degree at most $\bigO{k \log{n} / \eps^2}$ in $\revG$
  (reversing the upper bound on the unweighted in-degrees in
  $\sparseG$), and the flow algorithm updates the residual graph along
  at most $\bigO{k \log{n} / \eps^2}$ paths before terminating, the
  maximum unweighted out-degree over all vertices never exceeds
  $\bigO{k \log{n} / \eps^2}$.

  Within the Ford-Fulkerson framework, we give a refined analysis that
  takes advantages of the auxiliary $(\vout,\rin)$ edges.  Call an
  out-vertex $\vout$ \defterm{saturated} if the auxiliary edge
  $(\vout,\rin)$ is saturated; that is, if $(\vout,\rin)$ is not in
  the residual graph.  Call an in-vertex $\vin$ \defterm{saturated} if
  the edge $(\vin,\vout)$ is saturated and $\vout$ is not saturated.
  (A vertex $\vout$ or $\vin$ is called \defterm{unsaturated} if it is
  not saturated.)  We modify the search for an augmenting path to
  effectively end when we first visit an unsaturated vertex $\vout$ or
  an unsaturated $\vin$.  If we visit an unsaturated $\vin$, then we
  automatically complete a path to $\rin$ via $\vout$.  If we find an
  unsaturated $\vout$, then we automatically complete a path to $\rin$
  via the edge $(\vout,\rin)$. It remains to bound the running time of
  this search. We first bound the number of saturated $\vout$'s.

  \begin{claims}
  \item \labelclaim{saturated-vouts} There are at most
    $\bigO{k / \eps}$ saturated $\vout$'s.
  \end{claims}
  Indeed, each saturated $\vout$ implies $\bigO{\log{n} / \eps}$ units
  of flow along $(\vout,\rin)$, and the flow is bounded above
  $\bigO{k \log{n} / \eps^2}$.

  Note that \refclaim{saturated-vouts} also implies there are at most
  $\bigO{k / \eps}$ $\vin$'s such that $\vout$ is saturated. The next
  claim bounds the total out-degree of saturated $\vin$'s.

  \begin{claims}
  \item \labelclaim{saturated-vin-out-degrees} The sum of out-degrees
    of saturated $\vin$'s is at most the amount of flow that has been
    routed to $\rin$.
  \end{claims}

  Indeed, the out-degree of a $\vin$ in the residual graph is bounded
  above by the amount of flow through $(\vin,\vout)$, since initially
  $(\vin,\vout)$ is the only outgoing edge from $\vin$.  Recall that
  if $\vin$ is saturated, then by definition $\vout$ is unsaturated.
  As long as $\vout$ is unsaturated, each unit of flow through
  $(\vin,\vout)$ goes directly to $\rin$ via the edge $(\vout,\rin)$,
  and can be charged to the total flow.

  We now apply the above two claims to bound the total running time
  for each search, as follows.
  \begin{claims}
  \item \labelclaim{local-rvc-search-length} Every (modified) search
    for an augmenting path traverses at most
    $\bigO{k^2 \log{n} / \eps^2}$ edges.
  \end{claims}
  We first observe that every vertex visited in the search, except the
  unsaturated vertex terminating the search, is either (a) a saturated
  $\vin$, (b) a saturated $\vout$, or (c) an unsaturated $\vin$ such
  that $\vout$ is saturated.  We will upper bound the number of edges
  traversed in each iteration based on the type of vertex at the
  initial point of that edge.  First, the amount of time spent
  exploring edges leaving (a) a saturated $\vin$ is, by
  \refclaim{saturated-vin-out-degrees}, at most the size of the flow
  at that point, which is at most $\bigO{k \log{n} / \eps^2}$.
  Second, consider the time spent traversing edges leaving either (b)
  a saturated $\vout$ or (c) an unsaturated $\vin$ such that $\vout$
  is saturated. By \refclaim{saturated-vouts}, there are at most
  $\bigO{k / \eps}$ such vertices, and each has out-degree at most
  $\bigO{k \log{n} / \eps^2}$. Thus we spend
  $\bigO{k^2 \log{n} / \eps^2}$ time traversing such edges. All
  together, we obtain an upper bound of $\bigO{k^2 \log{n} / \eps^2}$
  total edges per search.

  \refclaim{local-rvc-search-length} also bounds the running time for
  each iteration. The algorithm runs for at most
  $\bigO{k \log{n} / \eps^2}$ iterations before either finding an
  $(\tout,\rin)$-cut or concluding that the weight of the minimum
  $(\tout,\rin)$-cut, rescaled to the input scale of $G$, is at least
  a constant factor greater than $\vc$.  The total running time
  follows.
\end{proof}

We now present the overall algorithm for finding vertex $\root$-cuts with
small sink components. The algorithm combines \reflemma{local-apx-rvc}
with randomly sampling for a vertex $t$ in the sink component of an
approximately minimum $\root$-cut. In the following, we let $\outdegree{\root}$ denote the \emph{unweighted} out-degree pf $\root$ in $G$. 

\begin{lemma}
  \labellemma{apx-rvc-small-sink} Let $G = (V,E)$ be a directed graph
  with positive vertex weights. Let $\root \in V$ be a fixed root vertex.
  Let $\eps \in (0,1)$, $\vc > 0$ and $k \in \naturalnumbers$ be given
  parameters.  There is a randomized algorithm that runs in
  $\bigO{m + (n - \outdegree{\root}) k^2 \log{n}^3 / \eps^4}$ time and has
  the following guarantee.  If there is a vertex $\root$-cut of capacity
  at most $\vc$ and where the sink component has at most $k$ vertices,
  then with high probability, the algorithm returns a vertex
  $(\root,t)$-cut of capacity at most $\epsmore \vc$.
\end{lemma}

\begin{proof}
  Let $\optSink$ be the sink component of the minimum vertex $\root$-cut
  subject to $\sizeof{\optSink} \leq k$. Assume the capacity of the
  vertex in-cut of $\optSink$ is at most $\vc$ (since otherwise the
  algorithm makes no guarantees). Let
  $V' = V \setminus (\setof{\root} \cup \outneighbors{\root})$ and note that
  $\sizeof{V'} = n - 1 - \outdegree{\root}$.

  Suppose we had a factor-2 overestimate
  $\ell \in \bracketsof{\sizeof{\optSink}, 2 \sizeof{\optSink}}$ of
  the number of vertices in $\optSink$. We apply
  \reflemma{local-apx-rvc} with upper bounds $\vc$ on the size of the
  cut and $\ell$ on the number of vertices in the sink component,
  which returns a data structure that, with high probability, is
  correct for all queries. Let us assume the data structure is correct
  (and otherwise the algorithm fails). We randomly sample
  $\bigO{(n - \outdegree{\root}) \log{n} / \ell}$ vertices from $V'$. For
  each sampled vertex $t$, we query the data structure from
  \reflemma{local-apx-rvc}. Observe that if $t \in \optSink$, then the
  query for $t$ returns an $\root$-cut with capacity at most
  $\epsmore \vc$. With high probability we sample at least one vertex
  from $\optSink$, which produces the desired $\root$-cut. By
  \reflemma{local-apx-rvc}, the total running time to serve all
  queries is
  $\bigO{m + (n - \outdegree{\root}) \ell^2 \log{n}^3 / \eps^4}$.

  A factor-2 overestimate $\ell$ can be obtained by enumerating all
  powers of $2$ between $1$ and $2k$. One of these choices of $\ell$
  will be accurate and produce the minimum $\root$-cut with high
  probability. Note that the sum of
  $\bigO{(n - \outdegree{\root}) \ell^2 \log{n}^3 / \eps^4}$ over this
  range of $\ell$ is dominated by the maximum $\ell$. The claimed
  running time follows.
\end{proof}


\subsection{Rooted vertex mincut for large sink components}

\label{sample-rvc}

The third and final part (before the overall algorithm) is an
approximation for the rooted vertex cut that is well-suited for large
sink components.
\begin{lemma}
  \labellemma{apx-rvc-big-sink}
  Let $G = (V,E)$ be a directed graph with positive vertex
  weights. Let $\root \in V$ be a fixed root vertex. Let $\eps \in (0,1)$,
  $\vc > 0$, and $k \in \naturalnumbers$ be given parameters. There is
  a randomized algorithm that runs in
  \begin{math}
    \apxO{m + (n - \outdegree{\root}) \parof{n / \eps^2 + n^{1.5} / k}}
  \end{math}
  time and has the following guarantee. If there is a vertex $\root$-cut
  of capacity at most $\vc$ and where the sink component has at most
  $k$ vertices, then with high probability, the algorithm returns a
  vertex $(\root,t)$-cut of capacity at most $\epsmore \vc$.
\end{lemma}

\begin{proof}
  Let $\optSink$ be the sink component of the minimum $\root$-cut subject
  to $\sizeof{\optSink} \leq k$. We assume the capacity of the $\root$-cut
  induced by $\optSink$ is at most $\vc$. (Otherwise the output is not
  well-defined.) Let
  $V' = V \setminus \parof{\setof{\root} \cup \outneighbors{\root}}$ and note
  that $\sizeof{V'} < n - \outdegree{\root}$.

  We apply \reflemma{apx-rvc-sparsification} to produce the graph
  $\sparseG$. \reflemma{apx-rvc-sparsification} succeeds with high
  probability and for the rest of the proof we assume it was
  successful. (Otherwise the algorithm fails.)  We sample
  $\bigO{ \parof{n - \outdegree{\root}} \log{n} / k}$ vertices
  $t \in V'$. For each sampled $t$, we compute the minimum
  $(\root,t)$-vertex cut in $\sparseG$. With high probability, some
  $t$ will be drawn from the sink component of the true minimum
  $\root$-cut, in which case the minimum $(\root,t)$-cut in $\sparseG$
  gives an $\epsmore$-approximate $\root$-cut in $G$ (by
  \reflemma{apx-rvc-sparsification}).  We use the $\apxO{m + n^{1.5}}$
  time vertex-capacitated flow algorithm \cite{brand+20}. By
  \reflemma{apx-rvc-sparsification}, we have
  $m = \bigO{n k \log{n} / \eps^2}$. This gives the total running
  time.
\end{proof}

\subsection{Approximating the rooted and global vertex mincut}

\label{rvc-all-together}

We now combine the two parameterized approximation algorithms for
rooted vertex mincut to give the following overall algorithm for
rooted vertex mincut. This establishes one part of
\ref{thm:vertex-apx} concerning approximate rooted vertex cuts.

\begin{restatable}{theorem}{ApxRVC}
  \labeltheorem{apx-rvc} Let $\eps \in (0,1)$, let $G = (V,E)$ be a
  directed graph with polynomially bounded vertex weights, and let
  $\root \in V$ be a fixed root. A $\epsmore$-approximate minimum vertex
  $\root$-cut can be computed with high probability in
  $\apxO{m + n (n - \outdegree{\root}) / \eps^2}$ randomized time.
\end{restatable}

\begin{proof}
  The high-level approach is similar to \Cref{thm:edge-apx} for edge
  mincut -- we are balancing two algorithms for rooted vertex
  mincut, where one is better suited for small sink components,
  the second is better suited for large sink components. Both leverage
  the randomized sparsification lemma.  As before, with
  polylogarithmic overhead, we can assume access to values $\vc$ and
  $k$ that are within a factor 2 of the weight of the minimum $\root$-cut
  and the number of vertices in the sink component of the minimum
  $\root$-cut, respectively.  For a fixed choice of $k$ and $\vc$, we
  run the faster of two randomized algorithms, both of which would
  succeed with high probability when $k$ and $\vc$ are
  (approximately) correct.  The first option, given by
  \reflemma{apx-rvc-small-sink}, runs in
  $\apxO{(n - \outdegree{\root}) k^2 / \eps^3}$.  The second option, given
  by \reflemma{apx-rvc-big-sink}, runs in
  $\apxO{n(n-\outdegree{\root}) / \eps^2 + n^{1.5} (n - \outdegree{\root}) /
    k}$.  The overall running time is obtained by choosing $k$ to
  balance the running times.  For $k = \eps \sqrt{n}$, we obtain the
  claimed running time.
\end{proof}

Next we use the algorithm for rooted vertex mincut to obtain an
algorithm for global vertex mincut and establish
\refcorollary{apx-vc}.  \citet{hrg-00} showed that running times of
the form $(n-\outdegree{\root}) T$ for rooted mincut from a root $\root$
imply a randomized $n T$ expected time algorithm for global vertex
mincut.  \reftheorem{apx-rvc} gives a
$\apxO{m + n (n-\outdegree{\root}) / \eps^2}$ running time, so some
modifications have to be made to address the additional $\apxO{m}$
additive factor.  This establishes the remaining part of \ref{thm:vertex-apx}.

\begin{restatable}{corollary}{ApxVC}
  \labelcorollary{apx-vc} For all $\eps \in (0,1)$, a
  $\epsmore$-approximate minimum weight global vertex cut in a
  directed graph with polynomially bounded vertex weights can be
  computed with high probability in $\apxO{n^2 / \eps^2}$ expected
  time.
\end{restatable}
\begin{proof}
  Let $T = \apxO{n / \eps^2}$.  Let $w: V \to \preals$ denote the
  vertex weights, and let $W = \sum_{v \in V} \weight{v}$ be the total
  weight of the graph. Let $\vc$ denote the weight of the minimum
  global vertex cut. The algorithm samples
  $L = \bigO{W \log{n} / (W - \vc)}$ vertices $\root$ in proportion to
  their weight, and -- morally, but not actually -- computes the
  minimum $\root$-vertex cut for each sampled vertex $\root$ via
  \reftheorem{apx-rvc}. It returns the smallest cut found.

  For the sake of running time, we adjust the algorithm from
  \reftheorem{apx-rvc}. Recall that for a fixed root $\root$, and for
  each of a logarithmic number of values for $k$ and $\vc$, the
  algorithm from \reftheorem{apx-rvc} applies
  \reflemma{apx-rvc-sparsification} which reduces the graph to having
  $\apxO{n k / \eps^2}$ edges and rooted mincut $\apxO{k/\eps^2}$. For
  fixed $k$ and $\vc$, rather than rerun
  \reflemma{apx-rvc-sparsification} entirely for each $\root$ we
  sample, we execute most of it just once for all $\root$, and make
  local modifications for each different root $\root$. Referring to
  the algorithm given in the proof of
  \reflemma{apx-rvc-sparsification}, observe that the only step that
  directly mentions $\root$ is step \refstep{rvc-extra-vertices},
  which adds auxiliary vertices between $\root$ and each other vertex
  $v$. We move this step to the very end of the algorithm. (Here the
  vertex weight of auxiliary vertices is scaled down appropriately.)
  It is easy to see that the proof of
  \reflemma{apx-rvc-sparsification} still goes through (with minor
  rearrangement in the argument). The advantage is that, over all $L$
  roots $\root$, we now spend a total of $\bigO{m + n L}$ time, rather
  than $\bigO{m L}$.  Thereafter, the rest of the rooted mincut
  algorithm takes $\apxO{(n-\outdegree{\root}) T}$ per root
  $\root$. Note that $\apxO{(n-\outdegree{\root}) T}$ dominates the
  $\bigO{n}$ time required to complete the sparsification for each
  root.

  Consider a single root $\root$ sampled from $V$ in proportion to its
  weight. The expected running time to compute the minimum $\root$-cut is
  \begin{align*}
    \evof{\parof{n - \outdegree{\root}} T}
    =                           %
    n T - \frac{T}{W}\sum_{v \in V} \outdegree{v} \weight{v}
    \tago{=}
    n T - \frac{T}{W} \sum_{v \in V} \sum_{x \in \inneighbors{v}} \weight{x}
    \tago{\leq}
    n T \parof{1 - \frac{\vc}{W}}.
  \end{align*}
  Here, in \tagr, $\inneighbors{v}$ denotes the in-neighborhood of
  $v$. The equality is obtained by implicitly interchanging
  sums. \tagr is because for each $v$, the sum
  $\sum_{x \in \inneighbors{v}} \weight{x}$ is the weighted in-degree
  of $v$, and at least $\vc$.  Thus The overall expected running time over
  all the sampled roots is
  \begin{align*}
    \bigO{\prac{W \log{n}}{W - \vc} \evof{\parof{n - \outdegree{\root}}
    T}}
    =
    \bigO{\prac{W \log{n}}{W - \vc} \cdot nT \parof{1 -
    \frac{\vc}{W}}}
    =                           %
    \apxO{n^2 / \eps^2}.
  \end{align*}
  Meanwhile, when we sample $\bigO{W \log{n} / (W - \vc)}$ vertices in
  proportion to their weight, then with high probability, at least one
  sampled vertex lies outside the minimum global vertex cut. Such a
  vertex then leads to the minimum global vertex cut with high
  probability.
\end{proof}


\end{krq}

\subsection*{Acknowledgement}
     This project has received funding from the European Research
        Council (ERC) under the European Union's Horizon 2020 research
        and innovation programme under grant agreement No
        715672. Nanongkai was also partially supported by the Swedish
        Research Council (Reg. No. 2019-05622). Panigrahi was supported in part
        by NSF Awards CCF 1750140 and CCF 1955703.  Quanrud was supported in part by NSF grant CCF-2129816. Cen and Panigrahi  would like to thank Yu Cheng and Kevin Sun for helpful discussions at the initial stages of this project. Quanrud thanks Chandra Chekuri for helpful discussions and feedback.
\bibliographystyle{alpha}

\bibliography{refs}

\appendix

\section{Proof of  \reflemma{apx-rec-skeleton}}
\label{sec:proof skeleton}

\begin{proof}
  For ease of notation, we prove \reflemma{apx-rec-skeleton} except
  with $\epsmore$ replaced by $(1 + c \eps)$, for a fixed constant
  $c > 1$, in properties \ref{rec-skeleton-min-cut} and
  \ref{rec-skeleton-apx-cuts}. The constant factor can then be removed
  by decreased $\eps$ by a constant factor in the construction.

  Let $\tau = c_{\tau} \eps^2 \ec / (k \log n)$ for a sufficiently
  small constant $c_{\tau} > 0$. Decreasing $c_{\tau}$ by at most a constant factor, we can assume that $\lambda$ and $\eps \lambda / 2k$ are integer multiples of $\tau$. Let $G_0 = (V,E_0)$ be the reweighted subgraph obtained as follows.
  \begin{enumerate}
  \item \labelstep{rec-importance-sample} Randomly round the weight of
    each edge independently up or down to the nearest multiple $\tau$,
    preserving the weight of each edge in expectation.\footnote{That is, if the weight of an edge $e$ in $G$ is $x\cdot \tau+y$ for an integer $x\geq 0$ and any $0\leq y<\tau$, then the weight of $e$ is  $(x+1)\cdot \tau$ with probability $y/\tau$ and $x\cdot \tau$ otherwise.} Drop any edges with weight rounded down to $0$.
  \item \labelstep{rec-extra-edges} Add an edge of weight
    $\eps \lambda / 2k$ from the root $\root$ to every vertex
    $v \in V - \root$.
  \item \labelstep{rec-scale-down} Scale down all the edge weights by $\tau$.
  \end{enumerate}

  For each set $\Sink\subseteq V$, let $w_1(\Sink) = \incutsize{\Sink}{G}$ denote
  the weight of the in-cut at $\Sink$ in $G$. Let $w_2(\Sink)$ denote
  the randomized weight of the in-cut after randomly rounding in step
  \refstep{rec-importance-sample}.  Let $w_3(\Sink)$ denote the
  randomized weight of the in-cut of $\Sink$ after adding the
  auxiliary edges in \refstep{rec-extra-edges}.  The first claim
  analyzes the concentration of $w_2(\Sink)$ for all sets
  $\Sink \subseteq V$.

  \begin{claims}
  \item \labelclaim{rec-concentration} With high probability, for all
    $\Sink \subseteq V$,
    \begin{align*}
      \absvof{w_2(\Sink) - w_1(\Sink)} \leq \eps w_1(\Sink) + \frac{\eps \lambda
      \sizeof{\Sink}}{2 k}.
    \end{align*}
  \end{claims}
  The above claim consists of an upper and lower bound on $w_2(\Sink)$ for
  all $\Sink$.  We first show the lower bound on $w_2(\Sink)$ holds for all $\Sink$
  with high probability.  Fix $\Sink \subseteq V$.  $w_2(\Sink)$ is an
  independent sum with expected value $w_1(\Sink)$ and where each term in
  the sum is nonnegative and varies by at most $\tau$. By a variation
  of standard Chernoff
  inequalities\footnote{\label{footnote:additive-chernoff} Here we
    apply the following bounds (appropriately rescaled) which follow
    from the same proof as the standard multiplicative Chernoff bound.
    \begin{quote}
      \itshape Let $X_1,\dots,X_n \in [0,1]$ independent random
      variables. Then for all $\eps > 0$ sufficiently small and all
      $\gamma > 0$,
      \begin{align*}
        \probof{X_1 + \cdots + X_n \leq \epsless \evof{X_1 + \cdots + X_n} -
        \gamma} \leq e^{-\eps \gamma},
      \end{align*}
      and
      \begin{align*}
        \probof{X_1 + \cdots + X_n \geq \epsmore \evof{X_1 + \cdots + X_n} +
        \gamma} \leq e^{-\eps \gamma}.
      \end{align*}
    \end{quote}
  }, for any $\gamma \geq 0$, we have
  \begin{align*}
    \probof{w_2(\Sink) \leq \epsless  w_1(\Sink) - \gamma} \leq e^{- \eps \gamma /
    \tau}
    =                           %
    n^{- \gamma k \log{n} / c_{\tau} \eps \lambda},
  \end{align*}
  In particular, for $\gamma = \eps \lambda \sizeof{\Sink} / 2 k$, the RHS is at most
  \begin{math}
    n^{- c_0 \sizeof{\Sink}}
  \end{math}
  where $c_0$ is a constant under our control (via $c_{\tau}$). For
  large enough $c_0$, we can take the union bound over all sets of
  vertices.  This establishes that the lower bounds for $w_2(\Sink)$
  hold for all $\Sink$ with high probability. The upper bounds also
  hold with high probability by a symmetric argument.

  Now we analyze the in-cuts after step
  \refstep{rec-extra-edges}.
  Recall that for $\Sink \subseteq V$, $w_3(\Sink)$ denotes the weight of
  the in-cut of $\Sink$ after adding the auxiliary edges in
  \refstep{rec-extra-edges}.
  \begin{claims}
  \item \labelclaim{rec-padded-concentration} With high
    probability, for all $\Sink \subseteq V - \root$, we have
    \begin{align*}
      \epsless w_1(\Sink) \leq w_3(\Sink) \leq \epsmore w_1(\Sink) + \eps \lambda
      \sizeof{\Sink} / k.
    \end{align*}
  \end{claims}
  Indeed, we have
  $w_3(\Sink) = w_2(\Sink) + \eps \lambda \sizeof{\Sink} / 2k$ for all
  $\Sink \subseteq V - \root$. The additive term introduced by $h$ offsets
  the additive error in the lower bound on $w_2(\Sink)$ in
  \refclaim{rec-concentration}. This term also adds on to the additive
  error in the upper bound of \refclaim{rec-concentration} for a total
  of $\eps \lambda \sizeof{\Sink} / k$. Thus in the high probability
  event of \refclaim{rec-concentration}, we have the bounds described
  by \refclaim{rec-padded-concentration} for all $\Sink$.

  Henceforth, let us assume that the high probability event in
  \refclaim{rec-padded-concentration} holds. (Otherwise the algorithm
  fails.) We now argue that \refclaim{rec-padded-concentration} implies that properties \ref{rec-skeleton-small-cut}, \ref{rec-skeleton-min-cut} and \ref{rec-skeleton-apx-cuts} hold whp.


 {\bf \ref{rec-skeleton-min-cut}:}
  Since
  the weight of each $\root$-cut decreases by at most an
  $\epsless$-multiplicative factor, the minimum $\root$-cut $\optSink$ has weight at
  least $\epsless \incutsize{\optSink}{G}$. Meanwhile $\incutsize{\optSink}$
  increases by at most a $\parof{1 + 2 \eps}$-multiplicative
  factor. Thus $\incutsize{\optSink}$ remains a $\apxmore$-approximate
  minimum $r$-cut.  This establishes property
  \ref{rec-skeleton-min-cut}.

  {\bf  \ref{rec-skeleton-apx-cuts}:}
  Since again the weight
  of any $\root$-cut decreases by at most an $\epsless$-multiplicative
  factor, any $\alpha$-approximate minimum cut in the randomly
  reweighed graph $G_0$ is an $\apxmore \alpha$-approximate min $\root$-cut in $G$. This establishes property
  \ref{rec-skeleton-apx-cuts}.

  {\bf \ref{rec-skeleton-small-cut}:}
  Finally, we observe that before
  scaling in step \refstep{rec-scale-down}, all edge weights are at least $\tau$, so after scaling, all
  edge weights are at least 1. Meanwhile, after scaling, the minimum
  $r$-cut has weight at most
  $\bigO{\incutsize{T} / \tau} = \bigO{k \log{n} / \eps^2}$. This
  establishes property \ref{rec-skeleton-small-cut} and completes the
  proof.
\end{proof}

\section{Proof of \Cref{thm:tree-packing}}
\label{sec:packarb_details}

We start by defining the fractional arborescence packing problem.
Let $A_1,A_2,\ldots, A_N$ be the set of all $s$-arborescences. We represent a fractional packing of arborescences as a vector $x\in \mathds{R}^N$, where $x_i$ represents the fractional contribution of arborescence $A_i$ in the packing.
Define the {\em value} of the arborescence packing as the sum of the coordinates of $x$, i.e., $\val(x) := \sum_{i=1}^N x_i$.
The goal is to maximize $\val(x)$, under the constraint that each edge $j$ can be fractionally used only up to its weight $w(j)$.
This definition generalizes integral arborescence packing, where we pack a multiset of arborescences, under the constraint that each edge $j$ is used at most $w(j)$ times, and the value is the cardinality of the multiset.

It will be convenient to state the fractional arborescence packing problem in the framework of a standard packing problem, defined below:
\begin{definition}[Packing problem \cite{Young1995}]
  \label{def:packing}
  For convex set $P\subseteq \mathds{R}^n$ and nonnegative linear function $f:P\to \mathds{R}^m$, the packing problem aims to find $\gamma^*=\min_{x\in P}\max_{j\in[m]} f_j(x)$, i.e., the solution in $P$ that minimizes the maximum value of $f_j(x)$ over all $j$.  The {\em width} of the packing problem $(P, f)$ is defined as $\omega=\max_{j\in[m],x\in P}f_j(x)-\min_{j\in[m],x\in P}f_j(x)$.
\end{definition}
For the fractional arborescence packing problem, we have the following: $P=\{x\in \mathds{R}^N : \val(x)=1,x\ge0\}$ is the convex hull of all arborescences. The function $f:P\to \mathds{R}^m$ defines the total usage of each edge $j$ in the arboresence packing normalized by the edge weight, i.e.,  $f_j(x)= \frac{\sum_{i\in[N]: j\in A_i} x_i}{w(j)}$ for $x\in P$ and any edge $j\in [m]$; we call this the {\em load} of arborescence packing $x$ on edge $j$.
The packing problem has width 
\begin{align}
    \omega = \max_{j\in[m],x\in P}f_j(x)-\min_{j\in[m],x\in P}f_j(x) \le \max_{j\in[m]}\frac{1}{w_j}-0 = \frac{1}{w_{\min}}, \label{eq:width bound}
\end{align}
where $w_{\min}$ is the minimum edge weight.
The objective function is to minimize the maximum load: $\gamma^*=\min_{x\in P}\max_{j\in[m]}f_j(x)$.

For any $x\in P$ with maximum load $\gamma=\max_{j\in[m]} f_j(x)$, we can multiply it by $1/\gamma$ to get an arborescence packing with value $1/\gamma$ and maximum load 1.
Conversely, for any fractional arborescence packing $x\in\mathds{R}^N$ with $\val(x)=v$ where $f_j(x)\le1$ for all edges $j$, we have $x/v\in P$.
Therefore, it suffices to look for the vector $x\in P$ achieving the optimal value $\gamma^*$, and then scale the vector up by $1/\gamma^*$ to obtain the maximum arborescence packing.
If the value of maximum $\root$-arborescence packing is $\lambda^*$, then
the optimal value of the packing problem is:
\begin{align}
    \gamma^*=\frac{1}{\lambda^*}. \label{eq:gamma star}
\end{align}

Next we describe the packing algorithm (Figure 2 of \cite{Young1995}).  Maintain a vector $y\in \mathds{R}^m$, initially set to $y=\bf{1}$.  In each iteration, find $x=\arg\min_{x\in P}\sum_j y_j f_j(x)$, and then add $x$ to set $S$ and replace $y$ by the vector $y'$ defined by $y'_j=y_j(1+\epsilon f_j(x))/\omega)$.  After a number of iterations, return $\bar{x}\in P$, the average of all the vectors $x$ over the course of the algorithm. The lemma below upper bounds the number of iterations that suffice:

\begin{lemma}[Corollary 6.3 of \cite{Young1995}]
  \label{lem:packing}
  After $\lceil\frac{(1+\epsilon)\omega\ln m}{\gamma^*((1+\epsilon)\ln(1+\epsilon)-\epsilon)}\rceil$ iterations of the packing algorithm, $\bar{\gamma}:=\max_j f_j(\bar{x})\le (1+\epsilon)\gamma^*$.
\end{lemma}

We will also make use of the (exact) duality between $s$-arborescence packing and minimum $s$-cut:

\begin{lemma}[Corollary 2.1 of \cite{Gabow1995}]
  \label{lem:duality}
  The value of maximum $\root$-arborescence packing is equal to the value of $\root$-mincut.
\end{lemma}

We now prove \Cref{thm:tree-packing}.

\begin{proof}[Proof of \Cref{thm:tree-packing}]
  First, construct $G_0$ according to \Cref{sparsification}, with edge size $m_0=\min\{m,O(nk\log n/\epsilon^2)\}$. Let $\lambda_0$ denote the minimum $s$-cut value on $G_0$. \Cref{sparsification} guarantees
 \begin{align}
    w_{\min}\ge 1 \mbox{ and } \lambda_0 = O(k\log n/\epsilon^2). \label{eq:lambda 0 upper bound}
 \end{align}
By the duality (\Cref{lem:duality} and \eqref{eq:gamma star}),
the optimal value of the packing problem formulated above is  $\gamma^*=\frac{1}{\lambda_0}$.
%
Run the aforementioned arborescence packing algorithm. \Cref{lem:packing}, \eqref{eq:width bound},\eqref{eq:gamma star} and \eqref{eq:lambda 0 upper bound} guarantee that after
$$O\left(\frac{\omega \ln m}{\gamma^*}\right)=O\left(\frac{\lambda_0\ln m}{w_{\min}}\right)= O(\lambda_0\ln m)= O(k \log^2 m)$$
iterations (with constant $\epsilon$) we have
\begin{align}
    \bar{\gamma}&\le(1+\epsilon)\gamma^*.
\end{align}
Then $\bar{x}/\bar{\gamma}$ is an arborescence packing with value $1/\bar{\gamma}\ge \frac1{1+\epsilon}\lambda_0$.
%
 Consider sampling a random $s$-arborescence $A$ from distribution $\bar x$, so we choose arborescence $A_i$ with probability $\bar x_i$. Since $\delta_{G_0}^-(\optSink) \le (1+\epsilon)\lambda_0$ by \Cref{sparsification},  the expected number of edges in $A\cap\partial_{G_0}^-(\optSink)$ is at most $(1+\epsilon)\lambda_0 / (\frac{1}{1+\epsilon}\lambda_0) = (1+\epsilon)^2\le 1+3\epsilon$ for small enough constant $\epsilon$.  Since we always have $|A\cap\partial_{G_0}^-(\optSink)|\ge1$, by Markov's inequality $\Pr[|A\cap\partial_{G_0}^-(\optSink)|-1\ge 1]\le 3\epsilon \le 1/2$ for small enough constant $\epsilon$.  Therefore, if we uniformly sample $\Theta(\log n)$ arborescences from the distribution $\bar x$, at least one of the arborescences is 1-respecting whp.

  It remains to compute $x=\arg\min_{x\in P}\sum_j y_j f_j(x)$ on each iteration. Since $\sum_jy_jf_j(x)$ is linear in $x$, the minimum must be achieved by a single arborescence.  So the task reduces to computing the minimum cost spanning $s$-arborescence, which can be done in $O(m_0+n\log n)$ \cite{Gabow1986}. The total time complexity, over all iterations, becomes $O((m_0+n\log n)\lambda_0\log n) = O(\min\{mk\log^2 n,nk^2\log^3 n\})$.  The construction of \Cref{sparsification} costs another $O(m\log n)$ time.
\end{proof}


\section{Proof of \reflemma{apx-rvc-sparsification}}
\label{proof-apx-rvc-sparsification}

\begin{algorithm}[t]
  \caption{Sparsifying a graph to preserve the rooted vertex
    mincut\label{vertex-sparsification-algorithm}}

Let
  $\tau = c_{\tau} \eps^2 \vc / k \log{n}$ and let
  $\Delta = c_{\Delta} k \log{n} / \eps^2$, where $c_{\tau} > 0$
  is a sufficiently small constant and $c_{\Delta} > 0$ is a
  sufficiently large constant.    \labelstep{rvc-first-step}

  Randomly round each vertex
  weight to (nearest) multiples of $\tau$. \labelstep{rvc-importance-sample}

For each vertex $v$,
  introduce an auxiliary vertex $a_v$ with weight
  $\eps \vc / 2 k$. Add edges from the $\root$ to $a_v$, and from
  $a_v$ to $v$.    \labelstep{rvc-extra-vertices}

  \tcc{Decreasing $c_{\tau}$ and $\eps$ as needed, we
    assume that $\vc$ and $\eps \vc / 2k$ are multiples of
    $\tau$.}

  Remove all outgoing edges from any vertex with weight $0$.

    Scale down all vertex weights by $\tau$ (which makes them integers).  \labelstep{rvc-rescale}

Truncate all vertex weights to be
  at most $c_{w} k \log{n} / \eps^2$ for a sufficiently large
  constant $c_w > 0$.   \labelstep{rvc-truncate}

For all $v$ with unweighted in-degree at
  least $\Delta$, replace all incoming edges to $v$ with a single
  edge from $\root$.   \labelstep{rvc-contract}
\end{algorithm}

\begin{proof}
  Consider the Algorithm \ref{vertex-sparsification-algorithm} applied
  to the input graph $G$.  Let $\sparseG$ be the graph obtained at the
  end of the algorithm. Properties \ref{rvcs-root} through
  \ref{rvcs-zero} follow directly from the construction. The remaining
  proof is dedicated to proving the high probability bounds of
  \ref{rvcs-cut-weights} and \ref{rvcs-small-cut-sets}. We first show
  that the initial steps \Cref{step:rvc-first-step} to \refstep{rvc-extra-vertices}
  -- before rescaling -- preserves the
  minimum weight rooted vertex-cut approximately in the sense of
  \ref{rvcs-cut-weights} (sans scaling). We then analyze the remaining
  steps.  For each set $S$, let $w_1(S)$ denote the weight of the vertex
  in-cut of $S$. Let $w_2(S)$ denote the randomized weight of the vertex
  in-cut after step \refstep{rvc-importance-sample}. Let $w_3(S)$ denote
  the randomized weight of the vertex in-cut after step
  \refstep{rvc-extra-vertices}.
  \begin{claims}
  \item
    \labelclaim{rvc-concentration}
    With high probability, for all $S \subseteq V$, we have
    \begin{align*}
      \absvof{w_2(S) - w_1(S)} \leq \eps w_1(S) + \frac{\eps \vc \sizeof{S}}{2k}.
    \end{align*}
  \end{claims}
  The claim and proof are similar to \refclaim{rec-concentration} in
  the proof of \reflemma{apx-rec-sparsification}. The claim consists
  of an upper bound and a lower bound on $w_2(S)$ for all $S$ and we
  first show the lower bound holds with high probability.  Fix any set
  $S$.  $w_2(S)$ is an independent sum with expected value $w_1(S)$ and
  where each term in the sum is nonnegative and varies by at most
  $\tau$. Concentration bounds (see footnote
  \ref{footnote:additive-chernoff} on page
  \pageref{footnote:additive-chernoff}) imply that for any
  $\gamma \geq 0$, we have
  \begin{align*}
    \probof{w_2(S) \leq \epsless  w_1(S) - \gamma} \leq e^{- \eps \gamma /
    \tau}
    =                           %
    n^{- \gamma k \log{n} / c_{\tau} \eps \vc},
  \end{align*}
  In particular, for
  $\gamma = \eps \vc \sizeof{S} / 2 k$, the RHS is at most
  \begin{math}
    n^{- c_0 \sizeof{S}}
  \end{math}
  where $c_0 > 0$ is again a constant under our control (via
  $c_{\tau}$). For sufficiently large $c_0$, we can take the union
  bound over all sets $S$, establishing the high probability lower
  bound. The high probability upper bound follows by a symmetric
  argument.

  Now we analyze the vertex $\root$-cuts after step
  \refstep{rvc-extra-vertices}.
  Recall that for $S \subseteq V$, $w_3(S)$ denotes the weight
  of the in-cut of $S$ after adding auxiliary vertices in step
  \refstep{rvc-extra-vertices}.
  \begin{claims}
  \item \labelclaim{rvc-padded-concentration} For all
    $S \subseteq V'$, we have
    \begin{align*}
      \epsless w_1(S) \leq w_3(S) \leq \epsmore w_1(S) + \frac{\eps \vc\sizeof{S}}{k}.
    \end{align*}
  \end{claims}
  This claim and its proof is similar to
  \refclaim{rec-padded-concentration} in
  \reflemma{apx-rec-sparsification}.  We have
  $w_3(S) = w_2(S) + \eps \vc \sizeof{S} / 2 k$ for all
  $S \subseteq V'$.  The additive factor of
  $\eps \vc \sizeof{S} / 2 k$ combine with the high-probability
  additive error in \refclaim{rec-concentration} to establish the
  claim.

  We point out that \refclaim{rvc-padded-concentration} implies that,
  with high probability after step \refstep{rvc-extra-vertices}, the
  weights of all the vertex $\root$-cuts are preserved in the approximate
  sense described by \ref{rvcs-cut-weights} (without the scaling).
  Henceforth we assume that the high probability event of
  \refclaim{rvc-padded-concentration} holds.  Now, after step
  \refstep{rvc-extra-vertices}, all the weights are integer multiples
  of $\tau$. We have $\vc / \tau = \bigO{k \log{n} / \eps^2}$.  After
  scaling down by $\tau$ in step \refstep{rvc-rescale}, we still
  preserve the $\root$-cuts per property
  \ref{rvcs-cut-weights}. Truncating weights in $\sparseG$ to
  $\bigO{\vc / \tau}$ decreases the weight of some cuts, but to no
  less than $\bigO{\vc / \tau}$ (which maps to $\bigO{\vc}$ when
  rescaled back to the scale of $G$). Removing the outgoing edges of
  vertices with weight $0$ also has no impact on the weight of any
  vertex $\root$-cut. The final step adding edges from $\root$ only eliminates
  some of the vertex $\root$-cuts from consideration and does not impact
  the weight of the remaining vertex $\root$-cuts. This establishes
  \ref{rvcs-cut-weights}.

  It remains to prove \ref{rvcs-small-cut-sets} and in particular we
  must show that it is not impacted by the last step,
  \refstep{rvc-contract}.  Recall that step \refstep{rvc-contract}
  replaces the incoming edges to any vertex $v$ with unweighted
  in-degree greater than $\Delta = \bigO{k \log{n} / \eps^2}$ with a
  single edge in $\root$.  In particular, this edge places $v$ in
  $\outneighbors{\root}$ and destroys all $\root$-cuts where the sink
  component contains $v$.  Let $\Sink \subset V'$ be the sink
  component of a vertex $\root$-cut in $G$ where the capacity of the cut
  is at most $\vc$, and $\sizeof{\Sink} \leq k$.  We want to show that
  all vertices in $\Sink$ have in-degree less than $\Delta$, in which
  case the extra edges in \refstep{rvc-contract} have no impact on
  $\Sink$.  The vertex in-cut induced by $\Sink$ has weight at most
  $(1 + 2 \eps) \vc / \tau$ in the randomized graph before
  \refstep{rvc-contract} (per \ref{rvcs-cut-weights}).  Fix any
  $v \in \Sink$ and consider the edges going into $v$. At most $k - 1$
  of those edges can come from another vertex in $\Sink$, since
  $\Sink$ has at most $k$ vertices. The remaining edges must be from
  vertices in the vertex in-cut of $\Sink$.  Each of these vertices
  have weight at least $1$, and by \ref{rvcs-cut-weights} the in-cut
  has weight at most $(1 + 2 \eps) \vc / \tau$, so there are at most
  $\bigO{\vc / \tau}$ of these vertices.  This gives a maximum total
  of less than $\Delta$ edges incident to $v$, as desired.  In
  conclusion, for any vertex $v$ that is the endpoint to at least
  $\Delta$ edges, it is safe to replace all of $v$'s incoming edges
  with a single edge from the root, without violating
  \ref{rvcs-small-cut-sets}. This establishes
  \ref{rvcs-small-cut-sets} and completes the proof.
\end{proof}





  
  

\end{document}